%% file: entreg.tex
\documentclass{amsart}
  \usepackage{amsmath}\interdisplaylinepenalty=2500
    \usepackage{amsfonts,amssymb}\usepackage{xspace}
      \usepackage{tikz,color,graphicx}
  \newfont{\mit}{cmmi10 at 11pt}    
  \DeclareMathAlphabet{\mathsfsl}{OT1}{cmss}{m}{sl}
   \newtheorem{theorem}{Theorem}          \newtheorem{corollary}{Corollary}
   \newtheorem{lemma}{Lemma}     
   \theoremstyle{remark}
   \newtheorem{remark}{Remark} \newtheorem{example}{Example}
\begin{document}
\title{Entropy region and convolution}
\author{Franti\v{s}ek Mat\'{u}\v{s} and L\'{a}szlo Csirmaz}
\thanks{The work of FM was partially supported by Grant Agency of
        the Czech Republic under Grant 13-20012S. The work of LCs
        was partially supported by TAMOP-4.2.2.C-11/1/KONV-2012-0001
        and by the Lendulet Program of the Hungarian Academy of Sciences.
        \\\indent
        The authors would like to express their gratitude to The Chinese
        University of Hong Kong for kind hospitality during First Workshop
        on Entropy and Information Inequalities, April 15-17, 2013,
        and fruitful discussions with participants.}
\address{Franti\v{s}ek Mat\'{u}\v{s} is with Institute of Information Theory and Automation,
            Academy of Sciences of the Czech Republic,
            Pod vod\'{a}renskou v\v{e}\v{z}\'{\i}~4, 182~08 Prague,
            Czech Republic (e-mail: matus@utia.cas.cz).}
\address{L\'{a}szlo Csirmaz is with Central European University,
            N\'{a}dor utca 9, H-1051, Budapest, Hungary,
            University of Debrecen, and Renyi Institute of Mathematics
            (e-mail: csirmaz@renyi.hu).}
\begin{picture}(0,0)
    \put(0,40){\makebox(0,0)[l]{\tiny submitted to
    \emph{Combinatorics, Probability and Computing}, \today}}
\end{picture}\vspace*{-5mm}
\begin{abstract}
    The entropy region is constructed from vectors of random variables by
    collecting Shannon entropies of all subvectors. Its shape is studied here
    by means of polymatroidal constructions, notably by convolution. The
    closure of the region is decomposed into the direct sum of tight and modular
    parts, reducing the study to the tight part. The relative interior of the reduction
    belongs to the entropy region. Behavior of the decomposition under selfadhesivity
    is clarified. Results are specialized to and completed for the region of four
    random variables. This and computer experiments help to visualize approximations
    of a symmetrized part of the entropy region. Four-atom conjecture on the
    minimization of Ingleton score is refuted.
\end{abstract}
\maketitle

%
 \renewcommand{\ge}{\geqslant}    \renewcommand{\le}{\leqslant}
 \renewcommand{\geq}{\geqslant}   \renewcommand{\leq}{\leqslant}
 \newcommand{\pmn}{\emptyset}\newcommand{\pdm}{\subseteq}\newcommand{\sm}{\setminus}
 \newcommand{\vynech}[1]{}
 \newcommand{\te}{\ensuremath{\tau}\xspace}\newcommand{\vte}{\ensuremath{\vartheta}\xspace}
   \def\R{\ensuremath{\mathbb{R}}\xspace}
   \def\F{\ensuremath{\mathbb{F}}\xspace}
   \def\Iij{\ensuremath{\mathbb{I}_{ij}}\xspace}
        \newcommand{\pomn}[1]{\ensuremath{\mathcal{P}(#1)}\xspace}
        \newcommand{\norm}[1]{\ensuremath{|\!|#1|\!|_\infty}\xspace}
        \newcommand{\vare}{\varepsilon}
 \newcommand{\HN}{\ensuremath{\boldsymbol{H}_{\!N}}\xspace}
 \newcommand{\Hf}{\ensuremath{\boldsymbol{H}}\xspace}
 \newcommand{\HNent}{\ensuremath{\boldsymbol{H}^{\scriptstyle\text{\sf ent}}_{\!N}}\xspace}
 \newcommand{\Hfent}{\ensuremath{\boldsymbol{H}^{\scriptstyle\text{\sf ent}}}\xspace}
        \newcommand{\Sec}{\ensuremath{\boldsymbol{S}}\xspace}
        \newcommand{\Lf}[1]{\ensuremath{\boldsymbol{L}_{#1}}\xspace}
        \newcommand{\Ff}[1]{\ensuremath{\boldsymbol{F}_{\!#1}}\xspace}
        \newcommand{\Ef}[1]{\ensuremath{\boldsymbol{E}_{\!#1}}\xspace}
 \newcommand{\HNsa}{\ensuremath{\boldsymbol{H}^{\scriptstyle\text{\sf sa}}_{\!N}}\xspace}
 \newcommand{\HNti}{\ensuremath{\boldsymbol{H}^{\scriptstyle\text{\sf ti}}_{\!N}}\xspace}
 \newcommand{\Hfti}{\ensuremath{\boldsymbol{H}^{\scriptstyle\text{\sf ti}}}\xspace}
 \newcommand{\HNmod}{\ensuremath{\boldsymbol{H}^{\scriptstyle\text{\sf mod}}_{\!N}}\xspace}
 \newcommand{\Hfmod}{\ensuremath{\boldsymbol{H}^{\scriptstyle\text{\sf mod}}}\xspace}
 \newcommand{\clHNent}{\ensuremath{\cl{\HNent}}\xspace}
 \newcommand{\clHNentti}{\ensuremath{\cl{\HNent}^{\scriptstyle\text{\sf ti}}}\xspace}
 \newcommand{\clHfent}{\ensuremath{\cl{\Hfent}}\xspace}
 \newcommand{\clHfentti}{\ensuremath{\cl{\Hfent}^{\scriptstyle\text{\sf ti}}}\xspace}
 \newcommand{\hti}{\ensuremath{h^{\scriptstyle\text{\sf ti}}}\xspace}
 \newcommand{\hmod}{\ensuremath{h^{\scriptstyle\text{\sf m}}}\xspace}
 \newcommand{\fti}{\ensuremath{f^{\scriptstyle\text{\sf ti}}}\xspace}
 \newcommand{\fmod}{\ensuremath{f^{\scriptstyle\text{\sf m}}}\xspace}
 \newcommand{\gti}{\ensuremath{g^{\scriptstyle\text{\sf ti}}}\xspace}
 \newcommand{\gmod}{\ensuremath{g^{\scriptstyle\text{\sf m}}}\xspace}
     \newcommand{\cl}[1]{\ensuremath{\mathsfsl{cl}(#1)}}
     \newcommand{\ri}[1]{\ensuremath{\mathsfsl{ri}(#1)}}
     \newcommand{\pol}[1]{\ensuremath{(#1)^{\circ}}}
    \newcommand{\HIng}{\ensuremath{\Hf^{\scalebox{0.6}{\slb}}}\xspace}
    \newcommand{\HIvij}{\HIng_{\!\scriptscriptstyle(ij)}}
    \newcommand{\HIngvio}[1]{\ensuremath{\HIng_{#1}}\xspace}
%
 \newcommand{\trn}[2]{{\scriptstyle\text{\mit\symbol{1}}}_{#2}%
    \def\temp{#1}\ifx\temp\empty\else{}\,\temp\fi}
 \newcommand\slb{\begin{tikzpicture}[scale=0.13\ht\strutbox]%
        \draw[line width=0.015cm](0,0)--(0.055,0.22)--(0.22,0.22);
        \draw[line width=0.035cm](0.22,0.22)--(0.165,0);
        \draw[line width=0.025cm](0.165,0)--(0,0);
                 \end{tikzpicture}}
\newcommand\stv[2]{\slb_{#2}%
    \def\temp{#1}\ifx\temp\empty\else\,\temp\fi}
    \renewcommand\ast{{*}}
    \newcommand\tsum{\mathop{\textstyle\sum}\nolimits}
    \newcommand{\bal}{\alpha}\newcommand{\bbe}{\beta}
    \newcommand{\bga}{\gamma}\newcommand{\bde}{\delta}

\section{Introduction}\label{S:intro}

 The entropy function of a random vector $(\xi_i)_{i\in N}$ with a finite
 index set~$N$ maps each $I\pdm N$ into the Shan\-non entropy of the sub\-vector
 $(\xi_{i})_{i \in I}$. When the vector takes finite number of values,
 the entropy function can be considered for a point of the Euclidean
 space $\R^{\pomn{N}}$ where $\pomn{N}$ is the power set of~$N$. Such
 points define the \emph{entropic region} \HNent. The closure \clHNent
 of the region is a convex cone \cite[Theorem~1]{ZhY.ineq.cond}
 whose relative interior is contained in \HNent \cite[Theorem~1]{M.twocon}.
 This work studies mostly the shape of \clHNent.

 Basic properties of the Shan\-non entropy imply that any entropy function
 $h$ from \HNent is non-decreasing and submodular, and thus the pair $(N,h)$
 is a polymatroid with the ground set $N$ and rank function $h$ \cite{Fuji}.
 The poly\-matroidal rank functions on~$N$ form the polyhedral cone \HN which,
 consequently, contains \clHNent. A polymatroid, or its rank function, is called
 \emph{entropic (almost entropic)} if the rank function belongs to \HNent
 (\clHNent).

 In this work, the entropy region and its closure are studied by means of
 standard constructions on polymatroids, recalled in Section~\ref{S:defs}. The
 central working tool is the \emph{convolution} of two poly\-matroidal rank
 functions. The crucial property is that \clHNent is closed under convolution
 with modular polymatroids \cite[Theorem~2]{M.twocon}. This has consequences
 on principal extensions and their contractions.

 In Section~\ref{S:dec}, the cone \clHNent is decomposed into the direct sum
 of two cones, see Corollary~\ref{C:tight+mod}. The first one consists of rank
 functions which give the same rank to $N$ and all subsets with one element less.
 We call them \emph{tight}. The second one is the cone of modular polymatroids,
 contained in the entropy region. Hence, the decomposition reduces the study of
 \clHNent to a cone of lesser dimension. It is also closely related to balanced
 information-theoretic inequalities \cite{Chan.bal}. The relative interior of
 the first cone is exhausted by entropic points, see Theorem~\ref{T:ri.ent}
 in Section~\ref{S:ri}.

 Section~\ref{S:self.adh} recalls the notion of selfadhesivity that describes
 amalgamation, or pasting, of copies of a polymatroid. It is the main ingredient
 in the majority of proofs of non-Shannon information-theoretical inequalities.
 The selfadhesivity is compared with the decomposition into tight and modular
 polymatroids. An alternative technique for proving inequalities is briefly
 discussed and related to principal extensions and their contractions.

 Starting from Section \ref{S:four-random} the set $N$ is assumed to have
 exactly four elements. The role of Ingleton inequality in the structure of
 \HN is recalled. The cone \clHNent is reduced to its subcone $\Lf{ij}$, cut
 off by a reversed Ingleton inequality and tightness. Applying polymatroidal
 constructions, it is shown that $\Lf{ij}$ is mapped by two linear maps
 to its face $\Ff{ij}$ of dimension~9, see Theorem~\ref{T:face}.

 Section \ref{S:symmetrized} investigates a symmetrization of $\Ff{ij}$
 whose cross-section $\Sec_{ij}$ has dimension three. Various numerical
 optimization techniques were employed to find an inner approximation
 of $\Sec_{ij}$. An outer approximation is compiled from available non-Shannon
 type information inequalities. Results are visualized, and indicate
 that the two approximations are far from each other. In Section
 \ref{S:4-atom}, the range of Ingleton score studied and related to the
 cross-section $\Sec_{ij}$, see Theorem~\ref{T:Istar}. In Example~\ref{Ex:L},
 a score is presented that refutes Four-atom conjecture \cite{Csir,DFZ.nonSh}.

 The concept of entropy region matters for several mathematical and engineering
 disciplines. The inequalities that hold for the points of the region are
 called information-theoretic, those that do not follow from the
 polymatroid axioms are frequently called non-Shannon. Main breakthroughs
 include finding of the first non-Shannon linear inequality by Zhang-Yeung
 \cite{ZhY.ineq} and the relation to group theory by \cite{Chan.Yeung}.
 The cone \clHNent is not polyhedral \cite{M.infinf} and the structure
 of non-Shannon inequalities seems to be complex \cite{Zh.gen.ineq,MMRV,
 DFZ,DFZ.nonSh,Csir.book,Walsh-Weber}. Reviews are
 in~\cite{Chan.progress,M.twocon} and elsewhere.

 In communications networks, the capacity region of multi-source network coding
 can be expressed in terms of the entropy region, the reader is referred to the
 thorough review of the network coding in~\cite{BMRST}. Non-Shannon inequalities
 have a direct impact on converse theorems for multi-terminal problems
 of information theory, see \cite{Y.course}. In cryptography, the inequalities
 can improve bounds on the information ratios in secret sharing schemes
 \cite{BeimelOrlov,BLP,Csir}.

 In probability theory, the implication problem of conditional independence
 among subvectors of a random vector can be rephrased via the entropy region, see \cite{St}.
 The guessing numbers of games on directed graphs and entropies of the graphs
 can be related to the network coding \cite{Riis.infoflows,Riis.constr's}
 where non-Shannon inequalities provide sharper bounds \cite{Riis.nonSh}.
 Information-theoretic inequalities are under investigation in additive
 combinatorics \cite{MMT}. Last but not least, the information-theoretic
 inequalities are known to be related to Kolmogorov complexity \cite{MMRV},
 determinantal inequalities and group-theoretic inequalities \cite{Chan.progress}.

\section{Preliminaries}\label{S:defs}

 This section recalls basic facts about polymatroids and related operations.
 Auxiliary lemmas are worked out to be used later. Introduction to entropy
 and the entropy region can be found in the textbooks \cite{CsiK,Y.course},
 further material on polymatroids is in \cite{Lov}.

 The letter $N$ always denotes a finite set and $f,g,h$ real functions
 on the power set $\pomn{N}$ of~$N$, or points in the $2^{|N|}$-dimensional
 Euclidean space $\R^{\pomn{N}}$. Singletons and elements of~$N$ are not
 distinguished and the union sign between subsets of~$N$ is often omitted.
 For example, $iJ$ abbreviates $\{i\}\cup J$ where $i\in N$ and $J\pdm N$.

 For $I\pdm N$ let $\delta_I\in\R^{\pomn{N}}$ have all coordinates
 equal to $0$ but $\delta_I(I)=1$. For $I,J\pdm N$ the expression
 $f(I)+f(J)-f(I\cup J)-f(I\cap J)$ is interpreted as the standard
 scalar product of $\trn{}{I,J}=\delta_I+\delta_J-\delta_{I\cup J}
 -\delta_{I\cap J}$ with $f$. An alternative notation for
 $\trn{}{iL,jL}$ is $\trn{}{ij|L}$ where $L\pdm N$ and $i,j\in N\sm L$.

\subsection{}
 The pair $(N,f)$ is a \emph{polymatroid} when $f(\emptyset)=0$, $f$
 is nondecreasing, thus $f(I)\leq f(J)$ for $I\pdm J\pdm N$, and submodular,
 thus $\trn{f}{I,J}\geq0$ for $I,J\pdm N$. Here, $N$ is the ground set,
 $f(N)$ is the rank and $f$ is the rank function of the polymatroid. The polymatroid
 is frequently identified with its rank function. The collection of polymatroidal
 rank functions forms the closed polyhedral cone \HN in the nonnegative orthant
 of $\R^{\pomn{N}}$. Extreme rays of the cone are mostly unknown.
 For a review of polymatroids the reader is referred to~\cite{Lov}.

 The polymatroid is a \emph{matroid} \cite{Oxl} if $f$ takes integer values and
 $f(I)\leq|I|$, $I\pdm N$.  For $J\pdm N$ and $0\leq m\leq|N\sm J|$ integer
 let $r^J_m(I)=\min\{m,|I\sm J|\}$, $I\pdm N$. Thus, $r^J_m$ is a matroidal
 rank function with the set of loops $J$, $r_m^J(J)=0$, and rank $m$. If
 $J=\pmn$ it is sometimes omitted in the superindex.

 The polymatroid $f$ is \emph{modular} if $\trn{f}{I,J}=0$ for any $I$ and $J$
 disjoint. This is equivalent to $f(I)=\sum_{i\in I}\,f(i)$, $I\pdm N$, or to
 the single of this equalities with $I=N$. The modular polymatroids form the
 polyhedral cone \HNmod whose extreme rays are generated by the matroids
 ${r_{\scriptscriptstyle1}}^{\scriptscriptstyle\!\!N\sm i}$, $i\in N$.

 A polymatroid $(N,f)$ is \emph{linear} if there exist subspaces $E_i$, $i\in N$,
 of a linear space over a field $\F$ such that if $I\pdm N$ then $f(I)$ equals
 the dimension of the sum of $E_i$ over $i\in I$. If \F is finite then
 $f\ln|\F|$ is entropic.

\subsection{}
 The \emph{contraction} of a polymatroid $(N,f)$ along $I\pdm N$ is $(N\sm I,h)$
 where $h(J)=f(J\cup I)-f(J)$, $J\pdm N\sm I$. (Poly)matroids are closed to
 contractions. The following lemma is known, e.g.\ implicit in the proof of
 \cite[Lemma 2]{M.back}, but no reference to a proof seems to be
 available.

\begin{lemma}\label{L:contract}
   The almost entropic polymatroids are closed to contractions.
\end{lemma}

\begin{proof}
 It suffices to show that if $f$ is equal to the entropy function of a random
 vector $(\xi_i)_{i\in N}$, then the contraction $h$ of $f$ along $I$ is
 almost entropic. If $\xi_i$ takes values in a finite set $X_i$ then
 $\xi_I=(\xi_i)_{i\in I}$ ranges in the product of $X_i$, $i\in I$. For
 every element $x_I$ of the product that is attained with a positive
 probability, let $\eta^{x_I}$ be the random vector $\xi_{N\sm I}$ conditioned
 on the event that $\xi_I=x_I$. The entropy function of $\eta^{x_I}$ is
 denoted by $g_{x_I}$. By an easy calculation, the contraction $h$ equals
 the convex combination of the entropy functions $g_{x_I}$ weighted
 by the probabilities of the events $\xi_I=x_I$. Since \clHNent is convex
 \cite[Theorem~1]{ZhY.ineq.cond}, $h$ is almost entropic.
\end{proof}

\subsection{}
 When $f$ and $g$ are polymatroidal rank functions on the same ground set $N$,
 their \emph{convolution} $f\ast g$ is defined as
 \[
    f\ast g\;(I) =\min_{J\pdm I}\:\big\{\, f(J)+g(I\sm J)\,\big\}\,,\qquad I\pdm N\,.
 \]
 If both $f$ and $g$ are modular, then $f\ast g$ is also modular, assigning the values
 $\min\{ f(i),g(i)\}$ to the singletons $i$ of $N$. By \cite[Theorem~2.5]{Lov},
 $(N, f\ast g)$ is a polymatroid whenever $g$ alone is modular. The following
 simple assertion may help to build intuition for later proofs.

\begin{lemma}\label{L:convolution}
   Let $f,g$ be two polymatroids on $N$ where $g$ is modular, and $i\in N$.
   If $f(j)\leq g(j)$ for all $j\in N\sm i$ then
   \[\begin{split}
      &f\ast g\:(I)=f(I)\,, \\
      &f\ast g\:(iI)=\min\:\big\{\,f(I)+g(i), f(iI)\,\big\}\,,
    \end{split} \qquad I\pdm N\sm i\,.
   \]
   If, additionally, $f(i)\leq g(i)$ then $f\ast g=f$.
\end{lemma}

\begin{proof}
 By submodularity of $f$, for $J\pdm I\pdm N\sm i$
 \[
    f(I)+g(\pmn)\leq f(J)+f(I\sm J) \leq f(J)+\mbox{$\sum$}_{j\in I\sm J} f(j)
                \leq f(J)+g(I\sm J)
 \]
 using that $f(j)\leq g(j)$, $j\in N\sm i$, and modularity of $g$. This proves
 that $f\ast g\:(I)$ equals $f(I)$. Similarly,
 \[
    f(iI)+g(\pmn)\leq f(iJ)+f(I\sm J)\leq f(iJ)+g(I\sm J)
 \]
 and
 \[
   f(I)+g(i)\leq f(J)+f(I\sm J)+g(i)\leq f(J)+g(iI\sm J)\,.
 \]
 Hence, $f\ast g\:(iI)$ is the smaller of the numbers $f(iI)$ and $f(I)+g(i)$.
\end{proof}

 In a notable instance of the convolution, the difference between
 $f\ast g$ and $f$ is at most at a singleton. This will be used in
 Theorem~\ref{T:face} to shift almost entropic points.

\begin{corollary}\label{C:cut}
    Let $(N,f)$ be a polymatroid, $i\in N$, and
    \[
        \max_{j\in N\sm i}\:\big[f(ij)-f(j)\big]\leq t \leq f(i)\,.
    \]
    Let $(N,g)$ be a modular polymatroid with $g(i)=t$ and $f(j)\leq g(j)$, $j\in N\sm i$.
    Then $f\ast g$ is equal to $f$ up to $f\ast g(i)=t$.
\end{corollary}

\begin{proof}
 The assumption $t\leq f(i)$ implies $f\ast g(i)=t$. By Lemma~\ref{L:convolution},
 $f\ast g$ is equal to~$f$ on the subsets of $N\sm i$. Let $I\pdm N\sm i$
 and $I$ contain some~$j$. By submodularity and the lower bound on $t$,
 \[
    f(iI)\leq f(I)+f(ij)- f(j)\leq f(I)+t=f(I)+g(i)\,.
 \]
 It follows by Lemma~\ref{L:convolution} that $f\ast g(iI)=f(iI)$.
\end{proof}

 Since the operation $\ast$ is commutative and associative, the convolution
 of a polymatroid $f$ with a modular polymatroid $g$ can be computed iteratively
 by Lemma~\ref{L:convolution}. It suffices to write $g$ as the multiple
 convolution of modular polymatroids $g_i$, $i\in N$, such that $g_i(i)=g(i)$
 and $g_i(j)=r$, $j\in N\sm i$, where $r$ is larger than the values of
 $f$ and $g$ on all singletons.

\subsection{}
 The last subsection of this section defines a one-element parallel extension
 of a polymatroid. This turns into a principal extension when modified
 by a convolution. Then, the added element is contracted. The polymatroid
 obtained by these three constructions is employed later in Sections~\ref{S:self.adh}
 and \ref{S:four-random}.

 Two points $i,j$ of a polymatroid $(N,f)$ are \emph{parallel} if $f(iJ)=f(jJ)$
 for every $J\pdm N$. Given any $i\in N$, it is always possible to extend
 $f$ to $\R^{\pomn{0\cup N}}$, where $0\not\in N$, such that $i$ and $0$
 are parallel in the extension. More generally, for $L\pdm N$ the
 \emph{extension of $f$ by $0$ parallel to $L$} is the polymatroid
 $(0\cup N,h)$ given by $h(J)=f(J)$ and $h(0\cup J)=f(L\cup J)$ where $J\pdm N$.
 If $f$ is the entropy function of $(\xi_i)_{i\in N}$ then $h$ is entropic
 as well, completing the random vector by the variable $\xi_0=(\xi_i)_{i\in L}$.

 This parallel extension is convolved with the modular polymatroid
 $(0\cup N,g)$ having $g(0)=t\leq f(L)$ and $g(i)\geq f(i)$, $i\in N$,
 to arrive at the \emph{principal extension $f_{L,t}$ of $f$ on the subset $L$
 with the value $t$} \cite{Lov}. By Lemma \ref{L:convolution},
 \[
   f_{L,t}(0\cup I)=\min\big\{f(I)+t,f(L\cup I)\big\}\,,\qquad I\pdm N\,.
 \]
 In turn, the principal extension $f_{L,t}$ is contracted by $0$ to get
 the polymatroid on $N$ with the rank function
 \begin{equation}\label{E:ff}
   f_{L,t}^*(I)=\min\big\{f(I),f(L\cup I)-t\big\}\,,\qquad I\pdm N\,.
 \end{equation}

\begin{lemma}\label{L:aent}
  If $(N,f)$ is almost entropic, $L\pdm N$ and $0\leq t\leq f(L)$
  then so is $(N,f_{L,t}^*)$.
\end{lemma}

\begin{proof}
 If $f$ is entropic then $f_{L,t}^*\in\clHNent$ by \cite[Theorem~2]{M.twocon}
 and Lemma~\ref{L:contract} which implies the assertion.
\end{proof}

 Under some assumptions, it is possible to find all mimina in~\eqref{E:ff}.
 Recall that the $f$-closure $\cl{I}$ of $I\pdm N$ consists of those $i\in N$
 that satisfy $f(iI)=f(I)$.  By monotonicity and submodularity, $f(\cl{I})=f(I)$.

\begin{lemma}\label{L:priext+contr2}
  If $(N,f)$ is a polymatroid, $L\pdm N$ and $0\leq t\leq f(L)$ such that
  \begin{equation}\label{E:aa}
    t\leq\;\; \min_{I\pdm N\,,\: L\not\pdm\cl{I}}\;\;
                        \max_{\ell\in L\sm\cl{I}}\;[\,f(\ell\cup I)-f(I)\,]
 \end{equation}
  then
  \[
    f_{L,t}^*(I)=\begin{cases}
                \;f(I)-t\,, \quad &\text{~~when~}L\pdm\cl{I}\,,\\
                \;f(I)\,,&\text{~otherwise,}
    \end{cases}\qquad I\pdm N\,.
 \]
\end{lemma}

\begin{proof}
 The assumption $0\leq t\leq f(L)$ is needed to derive \eqref{E:ff}. If
 $L\pdm\cl{I}$ then $f(L\cup I)\leq f(L\cup\cl{I})=f(\cl{I})=f(I)$. Hence,
 the inequality is tight and the minimum in~\eqref{E:ff} equals $f(I)-t$.
 Otherwise, by the assumption \eqref{E:aa}, $t\leq f(\ell\cup I)-f(I)$ for
 some $\ell\in L\sm\cl{I}$. Since $f(L\cup I)-t\geq f(\ell\cup I)-t\geq f(I)$
 the minimum in~\eqref{E:ff} equals~$f(I)$.
\end{proof}

\begin{remark}\label{R:qwe}
 The special instance of Lemma~\ref{L:priext+contr2} is used in~the proof
 of Theorem~\ref{T:face} to shift almost entropic points. There, $L$ equals
 a singleton $k$ contained in $\cl{N\sm k}$. In such a case, \eqref{E:aa} is
 a consequence of
 \begin{equation}\label{E:aaa}
    t\leq \min_{j\in N\sm k}\: [f(N\sm j)-f(N\sm jk)]
 \end{equation}
 because each maximum $f(kI)-f(I)$ in \eqref{E:aa} dominates the right-hand
 side of \eqref{E:aaa} by submodularity.
\end{remark}

 In another special instance $L=N$ of Lemma~\ref{L:priext+contr2}, the
 polymatroid $f_{L,t}^*$ is called the \emph{truncation} of $f$ by $t$,
 or to $f(N)-t$. It was applied e.g.\ in~\cite{Chan.linear}
 to investigate linear polymatroids.

\section{Decomposition into tight and modular polymatroids\label{S:dec}}

 The cone \HN of poly\-matroidal rank functions $h$ decomposes into the
 direct sum of the cone \HNti of \emph{tight} rank functions and the cone \HNmod
 of modular functions. Here, $h$ is tight if $h(N)=h(N\sm i)$, $i\in N$.
 The decomposition can be written as $h=\hti+\hmod$ where
 \[
    \begin{split}
    &\hti(I)=h(I)-\mbox{$\sum$}_{i\in I}\,[h(N)-h(N\sm i)]\,,  \\
    &\hmod(I)=\mbox{$\sum$}_{i\in I}\,[h(N)-h(N\sm i)]\,,
    \end{split}     \qquad I\pdm N\,.
 \]
 It is unique because the linear spaces $\HNti-\HNti$ and $\HNmod-\HNmod$
 intersect at the origin. In symbols, $\HN=\HNti\oplus\HNmod$.

\begin{theorem}\label{T:tight}
   If $h\in\clHNent$ then $\hti$ is almost entropic.
\end{theorem}

\begin{proof}
 Let $(N,h)$ be a polymatroid, $N'$ be a disjoint copy of $N$ and
 $i\mapsto i'$ a bijection between them. The polymatroid $(N,h)$
 extends to $(N\cup N',f)$ by
 \[
    f(I\cup J')=h(I\cup J)\,,\qquad I,J\pdm N\,,
 \]
 where $J'=\{j'\colon j\in J\}$. Thus, each $i'$ is parallel to $i$.
 Let $(N\cup N',g)$ be a modular polymatroid. Then, for $I\pdm N$
 \[
    f\ast g\,(I\cup N')=
        \min_{J\pdm I,\:K\pdm N}\:\big[\, h(J\cup K)+g(I\sm J)+g(N'\sm K')\,\big]\,.
 \]
 By monotonicity of $g$, the bracket does not grow when $K$ is replaced
 by $K\cup J$. Hence, the minimization can be restricted to the situations
 when $J=K\cap I$, and
 \[
    f\ast g\,(I\cup N')=
      \min_{K\pdm N}\:\big[\,h(K)+g(I\sm K)+g(N'\sm K')\,\big]\,.
 \]
 If $g(i)+g(i')=h(i)$ for $i\in N$ then
 \[
    f\ast g\,(I\cup N')=
      \min_{K\pdm N}\:
         \Big[\,h(K)+\mbox{$\sum$}_{i\in I\sm K}\,h(i)+g(N'\sm(I'\cup K'))\,\Big]\,.
 \]
 By submodularity of $h$, this minimization can be restricted to $K\supseteq I$, thus
 \begin{equation}\label{E:proof}
    f\ast g\,(I\cup N')=\min_{I\pdm K\pdm N}\:\big[\,h(K)+g(N'\sm K')\,\big]\,,
            \quad I\pdm N\,.
 \end{equation}

 In the case when
 \[
    g(i')=\hmod(i)=h(N)-h(N\sm i)\leq h(i)\,,\qquad i\in N\,,
 \]
 $h$ is decomposed to $\hti+\hmod$ and the minimum in \eqref{E:proof}
 is equal to $\hti(I)+\hmod(N)$. It is attained for $K=I$. It follows that
 $\hti(I)=f\ast g\,(I\cup N')- f\ast g\,(N')$. Hence, \hti is the
 contraction of $f\ast g$ along $N'$.

 If $h\in\clHNent$ then $f$ is almost entropic. The convolution theorem
 \cite[Theorem~2]{M.twocon} implies that $f\ast g\in\clHNent$.
 By Lemma \ref{L:contract}, $\hti\in\clHNent$.
\end{proof}

 The closure of the entropic region decomposes analogously to~\HN.
 As a consequence, the intersection $\clHNent\cap\HNti$
 is equal to
 \[
   \clHNentti=\{ \fti\colon f\in\clHNent\}\,.
 \]

\begin{corollary}\label{C:tight+mod}
   $\clHNent=[\clHNent\cap\HNti]\oplus\HNmod$.
\end{corollary}

\begin{proof}
 Theorem~\ref{T:tight} and $\HN=\HNti\oplus\HNmod$ imply the inclusion $\pdm$.
 The reverse one follows from the facts that \clHNent is a convex cone and
 $\HNmod\pdm\HNent$ \cite[Lemma~2]{M.twocon}.
\end{proof}

 It is open whether $\HNent$ equals $[\HNent\cap\HNti]\oplus\HNmod$.

\smallskip
 In the remaining part of this section it is shown that Corollary~\ref{C:tight+mod}
 is equivalent to \cite[Theorem 1]{Chan.bal} on balanced inequalities.

 Any nonempty closed convex cone $K$ in a Euclidean space is expressible as intersection
 of homogeneous closed halfspaces. This is reflected in the notion of the \emph{polar}
 cone $K^\circ$ of $K$ that consists of the outer normal vectors to $K$ at the origin,
 \[
    K^\circ=\big\{(\vte_I)_{I\pdm N}\in\R^{\pomn{N}}\colon\tsum_{I\pdm N}\:\vte_I h(I)\leq0
            \;\;\text{for all}\;\;h\in K\big\}\,,
 \]
 see \cite[Section~14]{Rock}. For example, the polar of \HNmod can be defined by the
 inequalities $\sum_{I\ni i}\vte_I\leq0$, $i\in N$; substituting
 ${r_{\scriptscriptstyle1}}^{\scriptscriptstyle\!\!N\sm i}$, $i\in N$, for $h$.
 The polars of $\HNent$ and $\clHNent$ coincide and are defined by
 the very linear information-theoretic inequalities.

 By \cite[Corollary 16.4.2]{Rock}, Corollary~\ref{C:tight+mod} is equivalent to
 \begin{equation}\label{E:inters}
    (\HNent)^\circ=(\clHNentti)^\circ\cap (\HNmod)^\circ\,.
 \end{equation}
 It was used tacitly also that $(\HNent)^{\circ\circ}=\clHNent$, and
 that \clHNentti, \HNmod and their sum are closed. The polar of $\clHNentti$
 consists of the vectors $(\vte_I)_{I\pdm N}$ satisfying
 $\mbox{$\sum$}_{I\pdm N}\:\vte_I\hti(I)\leq0$, $h\in\clHNent$, which
 rewrites to
 \begin{equation}\label{E:balance}
        \mbox{$\sum$}_{I\pdm N}\:\vte_I h(I)
           -\mbox{$\sum$}_{i\in N}\,[h(N)-h(N\sm i)]\mbox{$\sum$}_{I\ni i}\vte_I\leq0\,,
                \quad h\in\HNent\,.
 \end{equation}
 In turn, \eqref{E:inters} can be rephrased as \cite[Theorem 1]{Chan.bal}:
 \emph{Given $(\vte_I)_{I\pdm N}$, the inequality $\sum_{I\pdm N}\:\vte_I h(I)\leq0$
 holds for all $h\in\HNent$ if and only if \eqref{E:balance} is valid and
 $\sum_{I\ni i}\vte_I\leq0$, $i\in N$.}

 Writing, $\te_I=\vte_I$ when $|I|<|N|-1$,  $\te_{N\sm i}=\vte_{N\sm i}+\sum_{I\ni i}\vte_I$,
 $i\in N$, and $\te_N=\vte_N-\sum_{I\pdm N}|I|\vte_I$ the inequality in~\eqref{E:balance}
 rewrites to $\sum_{I\pdm N}\:\te_Ih(I)\leq0$. This one is \emph{balanced} in the sense
 $\sum_{I\ni i}\te_I=0$, $i\in N$. Thus, \eqref{E:balance} expresses all balanced
 information-theoretic inequalities.

\section{Entropy region: regular faces of \clHNent\label{S:ri}}

 As mentioned earlier, the relative interior of \clHNent
 belongs to the entropy region \HNent. Thus, \HNent and
 \clHNent differ only on the relative boundary of the
 latter. This section proves a stronger relation between
 them, motivated by the decomposition in Corollary 2.

\begin{theorem}\label{T:ri.ent}
   $\ri{\clHNentti}\oplus\HNmod\pdm\HNent\,$.
\end{theorem}

 The proof presented below is based on an auxiliary lemma.
 At the end of the section, faces of \clHNent are discussed.

\begin{lemma}\label{L:ri.ent}
   The cone \clHNentti contains a dense set of entropic points.
\end{lemma}

 A proof resorts to polymatroids constructed from groups. Recall
 that a polymatroid $(N,f)$ is \emph{group-generated} if there
 exists a finite group $G$ with subgroups $G_i$, $i\in N$, such
 that $f(I)=\ln|G|/|G_I|$ for  $I\pdm N$. Here, $G_I$ abbreviates
 $\bigcap_{i\in I}\:G_i$. Such  a polymatroid is always entropic.
 In fact, the group $G$ is endowed with the uniform probability
 measure and the polymatroid equals the entropy function of
 $(\xi_i)_{i\in N}$ where $\xi_i$ is the factor\-mapping of
 $G$ on the family $G/G_i$ of left cosets of $G_i$. The divisions
 of the group-generated polymatroidal rank functions by
 positive integers are dense in \clHNent \cite[Theorem~4.1]{Chan.Yeung}.

\begin{proof}[Proof of Lemma~\ref{L:ri.ent}]
 Given $\vare>0$ and $g\in\clHNentti$ there exists, by Theorem~\ref{T:tight},
 a random vector whose entropy function $h$ satisfies
 \[
    \norm{h-g}\triangleq\max_{I\pdm N}|h(I)-g(I)|\leq\vare\,.
 \]
 Since $g$ is tight
 \[\begin{split}
    \hmod(N)&=\hmod(N)-\gmod(N)\\
            &\leq \mbox{$\sum$}_{i\in N}\,
                |h(N)-g(N)|+|h(N\sm i)-g(N\sm i)|\leq2\vare|N|\,.
 \end{split}\]
 It is possible to assume that the vector sits on a finite set endowed with
 the uniform probability measure. By \cite[Remark~11]{M.twocon}, there exists
 a group $G$, group-generated polymatroid $f$ and integer $\ell\geq1$ such that
 $\norm{\frac{1}{\ell}f-h}\leq\vare$. Therefore,
 \[
    \tfrac{1}{\ell}\fmod(N)\leq|\tfrac{1}{\ell}\fmod(N)-\hmod(N)|+\hmod(N)\leq4\vare|N|\,.
 \]
 Let $(\xi_i)_{i\in N}$ be the corresponding random vector of factor\-mappings
 of $G$ onto $G/G_i$ whose entropy function equals $f$. If $I\pdm N$ then
 $\xi_I\triangleq(\xi_i)_{i\in I}$ takes $|G/G_I|$ values, each one with
 the same probability and $f(I)=\ln|G/G_I|$. Therefore, for every $j\in N$ there
 exists a random variable $\eta_j$ defined on $G$ such that it is constant on
 each coset of $G/G_N$, takes $|G_{N\sm j}/G_N|$ values and $(\xi_{N\sm j},\eta_j)$
 takes $|G/G_N|=|G/G_{N\sm j}||G_{N\sm j}/G_N|$ values. Necessarily, $\eta_j$ is
 a function of $\xi_N$, its entropy is $\ln|G_{N\sm j}/G_N|$, $\eta_j$ is
 stochastically independent of $\xi_{N\sm j}$, and they together determine
 $\xi_N$. Let $h'$ denote the entropy function of $(\zeta_i)_{i\in N}$ where
 $\zeta_i=(\xi_i,\eta_N)$ and $\eta_N=(\eta_j)_{j\in N}$. By construction,
 $h'(N\sm i)$ is the entropy of $(\xi_{N\sm i},\eta_{N})$, $i\in N$, and
 $h'(N)$ is the entropy of $(\xi_{N},\eta_{N})$. Hence, $h'$ is a tight
 entropy function. For $I\pdm N$
 \[
    f(I)\leq h'(I)\leq f(I)
        +\mbox{$\sum$}_{j\in N}\,\ln|G_{N\sm j}/G_N|
        =f(I)+\fmod(N)\,.
 \]
 It follows that $\norm{\frac{1}{\ell}h'-\frac{1}{\ell}f}
 \leq\frac{1}{\ell}\fmod(N)\leq4\vare|N|$. Hence,
 \[
    \norm{\tfrac{1}{\ell}h'-g}
        \leq\norm{\tfrac{1}{\ell}h'-\tfrac{1}{\ell}f}+\norm{\tfrac{1}{\ell}f-h}+\norm{h-g}
        \leq 4\vare|N|+2\vare\,.
 \]
 By \cite[Lemma~4]{M.twocon}, the tight polymatroid $\frac{1}{\ell}h'+\delta r_1$
 is entropic for any $\delta>0$. Thus, $\norm{(\tfrac{1}{\ell}h'+\vare r_1)-g}
 \leq 4\vare|N|+3\vare$ where $\vare$ can be arbitrarily small.
\end{proof}

\begin{remark}
  It is of separate interest that \clHNentti contains a dense set of points
  in the form $\frac{1}{m}h''$ where $h''$ is group-generated and $m\geq1$
  integer. In fact, the tight entropy function $h'$ from the previous proof
  need not be group-generated but arises from random variables defined on $G$
  with the uniform probability measure. Then, by \cite[Remark~11]{M.twocon},
  $h'$ can be arbitrarily well approximated by $\frac{1}{m}h''$ with $h''$
  group-generated. Since $h'$ is tight the construction of that remark
  provides $h''$ tight as well. Thus, to a given $g\in\clHNentti$ it is
  possible to construct $\frac{1}{\ell m}h''$ arbitrarily close,
  as in the above proof.
\end{remark}

\begin{proof}[Proof of Theorem~\ref{T:ri.ent}]
 Since, $\HNmod\pdm\HNent$ \cite[Lemma~2]{M.twocon} and $\HNent$ is closed
 to sums it suffices to prove that  $\ri{\clHNentti}\pdm\HNent\,$.
 Argumentation is analogous to that in the proof of \cite[Theorem~1]{M.twocon}.
 By~\cite[Lemma~3]{M.twocon}, the matroidal rank functions $r_1^J$ with $J\pdm N$ and $|J|<|N|-1$
 are linearly independent. Since they are tight and their nonnegative combinations
 are entropic they span a polyhedral cone contained $\HNti\cap\HNent$ whose dimension
 $2^{|N|}-|N|-1$ is the same as that of $\clHNentti$. Therefore, if $\epsilon>0$ then
 the set $B_{\epsilon}$ of polymatroids $\sum_{J\colon |J|<|N|-1}\,\alpha_J\, r_1^J$,
 where $0<\alpha_J<\epsilon$, is open in the linear space $\HNti-\HNti$ and the shifts
 of these sets provide a base for the relative topology.

 Hence, if $g$ belongs to the relative interior of $\clHNentti$ then it
 belongs to such a shift contained in the relative interior. It follows
 that $g-B_{\epsilon}$ is a subset of the relative interior for $\epsilon>0$
 sufficiently small. Since $g-B_{\epsilon}$ is a relatively open subset of
 $\clHNentti$ it contains an entropic polymatroid $h$, by Lemma~\ref{L:ri.ent}.
 This implies that $g$ can be written as $h+\sum_{J\colon |J|<|N|-1}\,\alpha_J r_1^J$
 where all $\alpha_J$ are nonnegative, and thus is entropic.
\end{proof}

 A convex subset $F$ of a convex set $K$ is a \emph{face} if it contains
 every line segment of $K$ which has an interior point in $F$. A face of
 a convex cone is a convex cone.

 Let us call a face $F$ of \clHNent \emph{regular} if all relative interior
 points of $F$ are entropic, thus $\ri{F}\pdm\HNent$. The trivial face
 $F=\clHNent$ is regular. Since the cones \clHNentti and \HNmod are defined
 by imposing tightness in monotonicity and submodularity, they are faces of
 \clHNent, by Theorem~\ref{T:tight}. Each face of \HNmod is a face of \clHNent,
 and is regular because $\HNmod\pdm\HNent$. By Theorem~\ref{T:ri.ent},
 the face $F=\clHNentti$ is regular.

\section{Selfadhesivity and tightness\label{S:self.adh}}

 This section recalls the notion of selfadhesivity and explores its
 relation to the decomposition $h=\hti+\hmod$ of polymatroids.
 The role of selfadhesivity in proving information-theoretic inequalities
 is briefly discussed and compared to an alternative technique by \cite{MMRV}.

 Two polymatroids $(N,h)$ and $(M,g)$ are \emph{adhesive} \cite{M.adh}, or adhere,
 if a polymatroid $(N\cup M,f)$ exists such that $f(I)=h(I)$ for $I\pdm N$,
 $f(J)=g(J)$ for $J\pdm M$, and $f(N)+f(M)=f(N\cup M)+f(N\cap M)$. Thus, the
 rank function $f$ is a common extension of $h$ and $g$, and the last equality
 expresses the adherence. A poly\-matroid $(N,h)$ is \emph{self\-adhesive at $O\pdm N$}
 if it adheres with the $\pi$-copy $(\pi(N),h_\pi)$ defined by a bijection
 $\pi\colon N\to\pi(N)$ such that $O=N\cap\pi(N)$, $\pi(i)=i$ for $i\in O$,
 and $h_\pi(\pi(I))=h(I)$, $I\pdm N$. A~poly\-matroid is
 \emph{self\-adhesive} if it is self\-adhesive at each $O\pdm N$.

 The rank functions of selfadhesive polymatroids on~$N$ form the polyhedral cone
 \HNsa \cite{M.adh}. This cone decomposes similarly to $\HN=\HNti\oplus\HNmod$.

\begin{theorem}\label{T:sa+tight}
    If $h\in\HNsa$ then $\hti$ is self\-adhesive.
\end{theorem}

\begin{proof}
 Let a polymatroid $(N,h)$ adhere with a $\pi$-copy at $O=N\cap\pi(N)$ and
 $\hat{N}=N\cup\pi(N)$. Thus, there exists an adhesive extension $(\hat{N},\hat{h})$.
 This extension is further extended to $(\hat{N}\cup\hat{N}',f)$, doubling each
 element of $\hat{N}$ by a parallel one in $\hat{N}'$, disjoint with $\hat{N}$.
 Similarly to the proof of Theorem~\ref{T:tight}, a modular polymatroid
 $(\hat{N}\cup\hat{N}',g)$ is constructed below such that the contraction of
 $f\ast g$ along $\hat{N}'$ witnesses that $(N,\hti)$ is selfadhesive at~$O$.

 The modular rank function $g$ is defined by
 \[\begin{split}
    &g(i)=g(\pi(i))=h(i)+h(N\sm i)-h(N)\,,\\
    &g(i')=g(\pi(i)')=h(N)-h(N\sm i)\,,
   \end{split}\qquad i\in N\,.
 \]
 Since $g(i)+g(i')=\hat{h}(i)$ and $g(\pi(i))+g(\pi(i)')=\hat{h}(\pi(i))$,
 an analogue of \eqref{E:proof} takes the form
 \begin{equation} \label{E:proof2}
   f\ast g(I\cup\hat{N}')
      =  \min_{I\pdm K\pdm\hat{N}}\:\big[\,\hat{h}(K)+g(\hat{N}'\sm K')\,\big]\,,
            \qquad I\pdm\hat{N}\,,
 \end{equation}
  arguing as in the proof of Theorem~\ref{T:tight}.

 If $i\in N\sm O$ then
 \[\begin{split}
    \hat{h}^{\scriptstyle\text{\sf m}}(i)= \hat{h}( \hat{N})- \hat{h}( \hat{N}\sm i)
        &= 2h(N)-h(O)-[h(N\sm i)+h(N)-h(O)]\\
        &=h(N)-h(N\sm i)=\hmod(i)=g(i')
  \end{split}
 \]
 because $\hat{h}$ is an adhesive extension of $h$ and $h_\pi$.  Analogously,
 if $i\in N\sm O$ then $\hat{h}^{\scriptstyle\text{\sf m}}(\pi(i))=\hmod(i)=g(\pi(i)')$.
 Therefore, the bracket in \eqref{E:proof2} rewrites to
 \[\begin{split}
    &\hat{h}^{\scriptstyle\text{\sf ti}}(K)
       +\hat{h}^{\scriptstyle\text{\sf m}}(K)
       +\hat{h}^{\scriptstyle\text{\sf m}}(\hat{N}\sm (O\cup K))+g(O'\sm K')\\
    &= \hat{h}^{\scriptstyle\text{\sf ti}}(K)
       +\hat{h}^{\scriptstyle\text{\sf m}}(\hat{N}\sm(O\sm K))
       +\hmod(O\sm K)\,.
 \end{split}\]
 Hence, the minimization in \eqref{E:proof2} can be further restricted
 to $K\pdm I\cup O$ and
 \[
   f\ast g(I\cup\hat{N}')
      = \hat{h}^{\scriptstyle\text{\sf m}}(\hat{N}\sm(I\cup O))+
        \min_{I\pdm K\pdm I\cup O}\:\big[\,\hat{h}(K)+\hmod(O\sm K)\,\big]\,,
            \quad I\pdm\hat{N}\,.
 \]

 The above minimum can be found in special cases. First,
 \[
    f\ast g(\hat{N}\cup\hat{N}')=\hat{h}(\hat{N})=2h(N)-h(O)
 \]
 using that $\hat{h}$ extends adhesively $h$ and its $\pi$-copy. Second,
 \[
   f\ast g(I\cup\hat{N}')
      = \hti(I)+\hmod(N)+\hmod(N\sm O)\,,
                \qquad I\pdm N\,,
 \]
 using that $\hat{h}(K)+\hmod(O\sm K)=\hti(K)+\hmod(I\cup O)$. Third,
 \[
   f\ast g(\pi(I)\cup\hat{N}')
      = \hti(I)+\hmod(N)+\hmod(N\sm O)\,,
                \qquad I\pdm N\,,
 \]
 by symmetry. It follows that the contraction of $f\ast g$ along $\hat{N}'$
 extends $\hti$ and its $\pi$-copy. The rank of the contraction is
 \[
    [2h(N)-h(O)]-[\hmod(N)+\hmod(N\sm O)]
        = 2\hti(N)-\hti(O)
 \]
 whence the extension is selfadhesive.
\end{proof}

\begin{corollary}\label{C:sa+tight}
    $\HNsa=[\HNsa\cap\HNti]\oplus\HNmod$.
\end{corollary}

\begin{proof}
 The inclusion $\pdm$ follows from Theorem~\ref{T:sa+tight}. Since the
 modular polymatroids have selfadhesive modular extensions and  \HNsa
 is a convex cone the opposite inclusion holds as well.
\end{proof}

 The convex cone \clHNent is not polyhedral \cite{M.infinf}, thus its polar cone
 is not finitely generated. There are infinite sets of linear information-%
 theoretic inequalities \cite{Zh.gen.ineq,MMRV,DFZ,DFZ.nonSh,Csir.book}, some
 of them rigorously proved and hundreds of them generated in computer experiments.
 The experiments are based on the fact that the entropic polymatroids are
 selfadhesive, $\HNent\pdm\HNsa$, and iterations of the idea. None of the
 experiments seems to have taken into account the possible reduction
 by imposing the tightness, cf.\ Corollary~\ref{C:sa+tight}.

\medskip
 A linear information-theoretic inequality
 \[
    \tsum_{I\pdm N}\:\vte_I h(I)\leq0 \;\;\text{for all}\;\;h\in\HNent
 \]
 is of non-Shannon type if $(\vte_I)_{I\pdm N}\in(\HNent)^\circ$ is not
 in $\HN^\circ$. There are two techniques for proving non-Shannon-type
 inequalities: either by selfadhesivity, as implicit in the original proof
 of Zhang-Yeung inequality \cite{ZhY.ineq.cond}, or alternatively by a
 lemma of Ahlswede and K\"{o}rner~\cite{CsiK}, as proposed in \cite{MMRV}.
 Recently it was found that the two techniques have the same power~\cite{Kaced.bal}.
 Actually, the original lemma from~\cite{CsiK} is not needed and only
 the following version on extensions suffices for proofs of~\cite{MMRV,Kaced.bal}.

\begin{lemma}\label{L:AhlK}
 If $(N,h)$ is almost entropic, $i\in N$ and $i'\not\in N$ then the polymatroid
 has an almost entropic extension $(i'\cup N,g)$ such that
 \[\begin{split}
    & g(i'\cup N\sm i)=g(N\sm i)\,,\\
    & g(i'\cup I)-g(i')=g(i\cup I)-g(i)\,,\quad I\pdm N\sm i\,.
    \end{split}
 \]
\end{lemma}

\begin{proof}
 The assumption implies that there exists an almost entropic and adhesive
 extension $(i'\cup N,f)$ of $(N,h)$ and its copy at $N\sm i$. Let
 $g$ denote the contraction $f_{L,t}^*$ of the principal extension
 $f_{L,t}$ of $f$ on the singleton $L=i'$ with $t=h(N)-h(N\sm i)$.
 By Lemma~\ref{L:aent}, $g$ is almost entropic. The value $t$ is at
 most $h(i)=f(L)$ whence \eqref{E:ff} applies and takes the form
 \[
   g(I)=\min\big\{f(I),f(i'\cup I)-h(N)+h(N\sm i)\big\}\,,\qquad I\pdm i'\cup N\,.
 \]
 If $I\pdm N\sm i$ then, using the properties of $f$ and submodularity,
 \[\begin{split}
     g(I)&=\min\{h(I),h(i\cup I)-h(N)+h(N\sm i)\}=h(I)\,,\\
     g(i\cup I)&=\min\{h(i\cup I),f(i'\cup i\cup I)-f(i'\cup N)+f(N)\}
        =h(i\cup I)\,,\\
     g(i'\cup I)&=h(i\cup I)-h(N)+h(N\sm i)\,.
 \end{split}\]
 The first and second equation show that $g$ is an extension of $h$. This and the
 last one imply $g(i'\cup N\sm i)=g(N\sm i)$ and $g(i'\cup I)-g(i')=g(i\cup I)-g(i)$.
\end{proof}

 The main ingredient in the above proof is a contraction of a principal extension,
 which relies on convolution. This indicates that selfadhesivity, convolution
 and other constructions on polymatroids seem to be powerful enough to rephrase all
 existing approaches to proofs of the linear information-theoretic
 inequalities.

\section{Entropy region of four variables}\label{S:four-random}

 This section presents more special applications of polymatroidal
 constructions and consequences of Theorems \ref{T:tight} and \ref{T:ri.ent}
 when the ground set $N$ has four elements. Results on reduction of \clHNent
 will be used later when minimizing Ingleton score. It is assumed that the
 four elements $i,j,k,l$ of $N$ are always different. In the notation
 for cones the subscript $N$ is omitted, for example $\Hf =\HN$.

 When studying the entropic functions of four variables
 the crucial role is played by the expression
 \[
       h(ik) + h(jk) + h(il) + h(jl) +  h(kl)
                  - h(ij) - h(k) -  h(l) -  h(ikl) - h(jkl)
 \]
 where $h\in\Hf$. It is interpreted also as the scalar product $\stv{}{ij}h$ of
 \[
    \stv{}{ij}=\delta_{ik}+\delta_{i\ell}+\delta_{jk}+\delta_{j\ell}+\delta_{k\ell}
                -\delta_{ij}-\delta_{k}-\delta_{\ell}-\delta_{ik\ell}-\delta_{jk\ell}
 \]
 with $h$. The inequality $\stv{}{ij}h\geq0$ holds when $h$ is linear, see
 the works of Ingleton \cite{Ingleton.ineq, Ingleton.rep}. Let \HIng denote the
 polyhedral cone of the functions $h\in\Hf$ that satisfy the six instances of the
 Ingleton inequality obtained by the permutation symmetry. By \cite[Lemma~3]{M.4var.I},
 \HIng has dimension $15$ and is generated by linear polymatroidal rank functions.
 Therefore, the functions from \ri{\HIng} are entropic due to~\cite[Theorem~1]{M.twocon}.

 By \cite[Lemma~4]{M.4var.I}, any $h\in\Hf\sm\HIng$ violates exactly one
 of the six Ingleton inequalities. Let $\HIvij$ denote the cone of
 functions $h\in\Hf$ with $\stv{h}{ij}\leq0$. It follows that \Hf is union
 of \HIng with the six cones $\HIvij,\ldots,\HIngvio{\scriptscriptstyle(kl)}$. Focusing
 primarily on $\clHfent$, it contains $\HIng$ and is contained in the union.
 By symmetry, it remains to study $\clHfent\cap\HIvij$.

 Let $\Lf{ij}$ denote the cone $\clHfentti\cap\HIvij$ of tight and almost entropic
 polymatroids $h$ that satisfy the reversed Ingleton inequality $\stv{}{ij}h\leq0$.

\begin{corollary}\label{C:4tight+mod}
  $\clHfent\cap\HIvij=\Lf{ij}\oplus\Hfmod\,$.
\end{corollary}

\begin{proof}
 Since the expression $\stv{}{ij}h$ is balanced, $\stv{}{ij}h=\stv{}{ij}\hti$
 and $\Hfmod$ is contained in $\HIvij$. This and Corollary~\ref{C:tight+mod}
 imply the decomposition.
\end{proof}

 The study of \clHfent thus reduces to that of $\Lf{ij}$.
 This cone is contained in $\HIvij\cap\Hfti$ which is known to be the
 conic hull of $11$ linearly independent polymatroidal rank functions
 \cite[Lemma~6.1]{M.4var.I}. The most notable one
 \begin{equation}\label{E:nonaent}
    \bar{r}_{ij}(K)=\begin{cases}\; 3\,,   &       K\in\{ ik, jk, il, jl, kl \},\\
                           \; \min  \{ 4, 2 |K|\}\,, &  \text{otherwise}
              \end{cases}
 \end{equation}
 is not almost entropic by Zhang-Yeung inequality \cite{ZhY.ineq.cond}. The
 remaining ones are matroidal
 \begin{equation}\label{E:linmat}
    \text{
            $r_{1}^\pmn$,\,\,\,\,\,
            $r_{3}^\pmn$,\,\,\,\,\,
            $r_{1}^{i}$, $r_{1}^{j}$,\,\,\,\,\,
            $r_{2}^{k}$, $r_{2}^{l}$,\,\,\,\,\,
            $r_{1}^{ik}$, $r_{1}^{jk}$, $r_{1}^{il}$, $r_{1}^{jl}$
         }
 \end{equation}
 where the matroids are uniform up to loops. Recall that
 the subindex denotes the rank and the superindex the set of loops.
 By the proof of \cite[Lemma~6.1]{M.4var.I}, every $g\in\HIvij\cap\Hfti$
 is the unique conic combination of the rank functions from \eqref{E:nonaent}
 and \eqref{E:linmat},
 \begin{equation}\label{E:comb}
    \begin{split}
    g=&-(\stv{g}{ij})\,\bar{r}_{ij}
        + (\trn{g}{ij|\emptyset})\,r_1  + (\trn{g}{kl|ij})\,r_3\\
      &\qquad +(\trn{g}{kl|i})\,r_{1}^{i} + (\trn{g}{kl|j})\,r_{1}^{j}
      + (\trn{g}{ij|k})\,r_{2}^{l}         + (\trn{g}{ij|l})\,r_{2}^{k}\\
      &\qquad +(\trn{g}{jl|k})\,r_{1}^{ik}  + (\trn{g}{il|k})\,r_{1}^{jk}
        + (\trn{g}{jk|l})\,r_{1}^{il} + (\trn{g}{ik|l})\,r_{1}^{jl}
    \end{split}
 \end{equation}
 identifying explicitly the coordinate functionals.

 Since the matroids in \eqref{E:linmat} are linear and there exists an
 entropic point violating Ingleton inequality, the dimension of $\Lf{ij}$
 is $11$, the same as that of \clHfentti or $\HIvij\cap\Hfti$.
 Theorem~\ref{T:ri.ent} has the following consequence.

\begin{corollary}\label{C:4ri.ent}
    $\ri{\Lf{ij}}\pdm\Hfent\,$.
\end{corollary}

\medskip
 The remaining part of this section focuses on faces of $\Lf{ij}$.
 The face $\Ff{ij}$ given by the equalities $\trn{g}{ij|\emptyset}=0$
 and $\trn{g}{kl|ij}=0$ plays a special role later, in particular when
 optimizing Ingleton score.

 Let $A_{i,j}$ and $B_{ij,k}$ be the linear mappings defined
 on $\R^{\pomn{N}}$ by
 \[
    A_{i,j}g=g+(\trn{g}{ij|\emptyset})\,(r_{1}^{i}-r_1)\quad\text{and}\quad
    B_{ij,k}g=g+(\trn{g}{kl|ij})\,(r_{2}^{k}-r_3)\,.
 \]

\begin{lemma}\label{L:AB}
    The mappings $A_{i,j}$ and $B_{ij,k}$ commute, leave invariant the
    hyperplanes given by $\trn{g}{ij|\emptyset}=0$ and $\trn{g}{kl|ij}=0$,
    respectively, $A_{i,j}$ maps into the first hyperplane, $B_{ij,k}$
    into the second one, and
    \[
        \stv{g}{ij}=\stv{(A_{i,j}g)}{ij}=\stv{(B_{ij,k}g)}{ij}\,,
        \qquad g\in\R^{\pomn{N}}\,.
    \]
\end{lemma}

\noindent
 A simple proof is omitted, for example
 \begin{equation}\label{E:ABg}
    A_{i,j}B_{ij,k}g=B_{ij,k}A_{i,j}g
      =g+(\trn{g}{ij|\emptyset})\,(r_{1}^{i}-r_1)+(\trn{g}{kl|ij})\,(r_{2}^{k}-r_3)\,.
 \end{equation}
 Both $A_{i,j}$ and $B_{ij,k}$ change at most two coordinates in~\eqref{E:comb}.

\begin{theorem}\label{T:face}   
  $A_{i,j}B_{ij,k}\,\Lf{ij}=\Ff{ij}\,$.
\end{theorem}

\begin{proof}
 The hyperplanes given by $\trn{g}{ij|\emptyset}=0$ and $\trn{g}{kl|ij}=0$
 peal out two facets of $\HIvij\cap\Hfti$, due to \eqref{E:comb}.
 By Lemma~\ref{L:AB},  $A_{i,j}B_{ij,k}$ maps $\HIvij\cap\Hfti$
 onto the intersection of the two facets. Since $\Lf{ij}$ is equal to
 $\HIvij\cap\Hfti\cap\clHfent$ it suffices to prove that both
 $A_{i,j}$ and $B_{ij,k}$ map $\Lf{ij}$ into \clHfent.

 By the identity
 \[
    \stv{}{ij}  =\trn{}{ij|k}+\trn{}{ik|l}+\trn{}{kl|j}-\trn{}{ik|j}\,,   
 \]
 if $f\in\HIvij$ then $\trn{f}{ik|j}\geq\trn{f}{ij|k}$, thus
 $f(ij)-f(j)\geq f(ik)-f(k)$. By symmetry, $f(ij)-f(j)\geq f(il)-f(l)$.
 In turn, Corollary~\ref{C:cut} can be applied to~$t=h(ij)-h(j)$,
 and provides $h\in\Hf$ that coincides with $f$ except at~$i$
 where $h(i)=f(ij)-f(j)$. Similarly, the rank functions $r^i_1$ and $r_1$
 differ only at $i$ and $r^i_1(i)-r_1(i)=-1$. It follows that $h=A_{i,j}f$.
 If additionally $f\in\clHfent$ then $h$, being the convolution of $f$ with
 a modular polymatroid, is almost entropic. Therefore, $f\in\Lf{ij}$ implies
 $A_{i,j}f\in\clHfent$.

 By the identity
 \[
    \stv{}{ij}=\trn{}{ij|k}+\trn{}{ik|l}+\trn{}{kl|ij}-\trn{}{ik|jl} \,,   
 \]
 if $f\in\HIvij$ then $\trn{f}{ik|jl}\geq\trn{f}{kl|ij}$. Additionally,
 if $f$ is tight this inequality rewrites to $f(ij)\geq f(jl)$. By symmetry,
 $f(ij)\geq f(il)$. It follows that \eqref{E:aaa} is valid for $t=f(N)-f(ij)$.
 By Remark~\ref{R:qwe} and $t\leq f(k)$, Lemma~\ref{L:priext+contr2} is applied
 with $L=k$ and provides $h=f^*_{k,t}$. This rank function differs from $f$
 by $t$ on the sets $I\pdm N$ with $k\in\cl{I}$. By the identity
 \[
    \stv{}{ij}  =\trn{}{ij|k}+\trn{}{ij|l}+\trn{}{kl|ij}-\trn{}{ij|kl}\,,      
 \]
 $\trn{f}{ij|kl}\geq\trn{f}{kl|ij}$. Since $f$ is tight the inequality rewrites
 to $f(ij)\geq f(kl)$. By symmetry, $f(ij)$ is maximal among all $f(J)$ with
 $|J|=2$. Therefore, if $t>0$ then $k\in\cl{I}$ is equivalent to $I\ni k$ or
 $I=N\sm k$. These are exactly the cases when $r^k_2$ and $r_3$ differ, and
 $r^k_2(I)-r_3(I)=-1$. It follows from $t=\trn{f}{kl|ij}$ that $h=B_{ij,k}f$.
 If, additionally, $f\in\clHfent$ then $h$ is almost entropic. Therefore,
 $f\in\Lf{ij}$ implies $B_{ij,k}f\in\Lf{ij}$.
\end{proof}

\begin{remark}\label{R:sppoint}
 Let $\Ef{ij}$ be the face of $\Lf{ij}$ given by the equalities
 \[
    \text{$\trn{g}{ij|k}=0$, $\trn{g}{ij|l}=0$, $\trn{g}{kl|i}=0$,
          $\trn{g}{kl|j}=0$ and $\trn{g}{kl|ij}=0$.}
 \]
 In \cite[Example 2]{M.4var.II}, four random variables are constructed such
 that their entropy function $g$ satisfies the above five constraints,
 $\stv{g}{ij}<0$, each of $\trn{g}{ij|\emptyset}$, $\trn{g}{jl|k}$,
 $\trn{g}{il|k}$, $\trn{g}{jk|l}$ and $\trn{g}{ik|l}$ is positive, and
 $g$ is not tight. The lack of tightness makes $g$ to be outside $\Lf{ij}$.
 Nevertheless, Theorem~\ref{T:tight} implies that $\gti\in\clHfent$ whence
 $\gti$ belongs to the face $\Ef{ij}$, even more, it belongs to its relative
 interior. At the same time, \cite[Theorem~4.1]{M.4var.III} implies that no
 point of $\ri{\Ef{ij}}$ is entropic. This phenomenon can be equivalently
 rephrased in terms of conditional information inequalities, studied recently
 in~\cite{KacRo,KacRo.non,KacRo.aent}.
\end{remark}

\section{Symmetrization of $\Ff{ij}$\label{S:symmetrized}}

 As before, the ground set $N$ is assumed to have four elements
 $i,j,k,l$, which are always different. In the previous section the study
 of \clHNent was reduced to that of $\Lf{ij}$, and a particular face
 $\Ff{ij}$ of the latter was identified. Here, a symmetrization of $\Ff{ij}$
 is described and its cross-section visualized, owing to numerical computer
 experiments.

 The expression $\stv{}{ij}$ and the cones $\Lf{ij}$ and $\Ff{ij}$ enjoy
 natural symmetries. Namely, if a permutation $\pi$ on~$N$ stabilizes
 the two-element set $ij$ then $\stv{h}{ij}=\stv{h_\pi}{ij}$, $h\in\Hf$.
 Hence $\Lf{ij}$ and $\Ff{ij}$ are closed to the action $h\mapsto h_\pi$.

 Let $C_{ij}$ be the linear mapping on $\R^{\pomn{N}}$ given by
 \[
    C_{ij}h\triangleq |G_{ij}|^{-1} \tsum_{\pi\in G_{ij}}\, h_\pi
 \]
 where $G_{ij}$ denotes the stabilizer of~$ij$, consisting of four permutations.
 By the decomposition \eqref{E:comb}, if $h\in\HIvij\cap\Hfti$ then
 \[\begin{split}
    C_{ij}h=&-(\stv{h}{ij})\,\bar{r}_{ij}
        + (\trn{h}{ij|\emptyset})\,r_1^\pmn  + (\trn{h}{kl|ij})\,r_3^\pmn\\
      &\qquad +\tfrac12\big[\trn{h}{kl|i} + \trn{h}{kl|j}\big] [r_{1}^{j}+r_{1}^{i}]\\
      &\qquad +\tfrac12\big[\trn{h}{ij|k} + \trn{h}{ij|l}\big] [r_{2}^{l}+r_{2}^{k}]\\
      &\qquad +\tfrac14\big[\trn{h}{jl|k}  + \trn{h}{il|k} + \trn{h}{jk|l} +\trn{h}{ik|l}\big][r_{1}^{ik}+r_{1}^{jk}+r_{1}^{il}+r_{1}^{jl}]\,.
    \end{split}
 \]
 It follows that $C_{ij}\Lf{ij}$ has dimension $6$ and $C_{ij}\Ff{ij}$
 is a face of dimension~$4$. The cross-section
 \[
    \Sec_{ij}\triangleq\{h\in C_{ij}\Ff{ij}\colon h(N)=1\}
 \]
 is three-dimensional. By \eqref{E:comb}, for $h\in\Sec_{ij}$
 \[\begin{split}
    1=h(N)=&\big[-4\stv{h}{ij}\big] + \big[\trn{h}{kl|i} + \trn{h}{kl|j}\big]
         + \big[2\trn{h}{ij|k}+2\trn{h}{ij|l}\big]\\
     & +\big[\trn{h}{jl|k}+\trn{h}{il|k}+\trn{h}{jk|l}+\trn{h}{ik|l}\big]\,.
    \end{split}
 \]
 Denoting by $\bar{\alpha}_h$, $\bar{\beta}_h$, $\bar{\gamma}_h$ and $\bar{\delta}_h$
 the above brackets, respectively, any function $h\in\Sec_{ij}$ takes the form
 \[
    h=\bar{\alpha}_h\,\tfrac14\bar{r}_{ij}+ \bar{\beta}_h\,\tfrac12[r_{1}^{j}+r_{1}^{i}]
            +\bar{\gamma}_h\,\tfrac14 [r_{2}^{l}+r_{2}^{k}]
            +\bar{\delta}_h \tfrac14 [r_{1}^{ik}+r_{1}^{jk}+r_{1}^{il}+r_{1}^{jl}]\,.
 \]
 Here, $\bar{\alpha}_h$, $\bar{\beta}_h$, $\bar{\gamma}_h$ and $\bar{\delta}_h$
 are nonnegative and sum to one. Further, $\bal=\tfrac14\bar{r}_{ij}$,
 $\bbe=\tfrac12[r_{1}^{j}+r_{1}^{i}]$, $\bga=\tfrac14 [r_{2}^{l}+r_{2}^{k}]$ and
 $\bde=\tfrac14 [r_{1}^{ik}+r_{1}^{jk}+r_{1}^{il}+r_{1}^{jl}]$ are linearly independent
 polymatroidal rank functions from $\Sec_{ij}$. It follows that $\Sec_{ij}$ is
 a closed convex subset of the tetrahedron with the vertices $\bal$, $\bbe$, $\bga$
 and $\bde$. Since the points $h$ having $\bar{\alpha}_h=0$ are almost entropic and
 $\bar{r}_{ij}$ is not, $\Sec_{ij}$ contains the triangle $\bbe\bga\bde$
 but not the vertex~$\bal$.

\begin{figure}[h!tb]
\begin{center}
\newsavebox\imagebox
\sbox\imagebox{\scalebox{0.0525}{\includegraphics{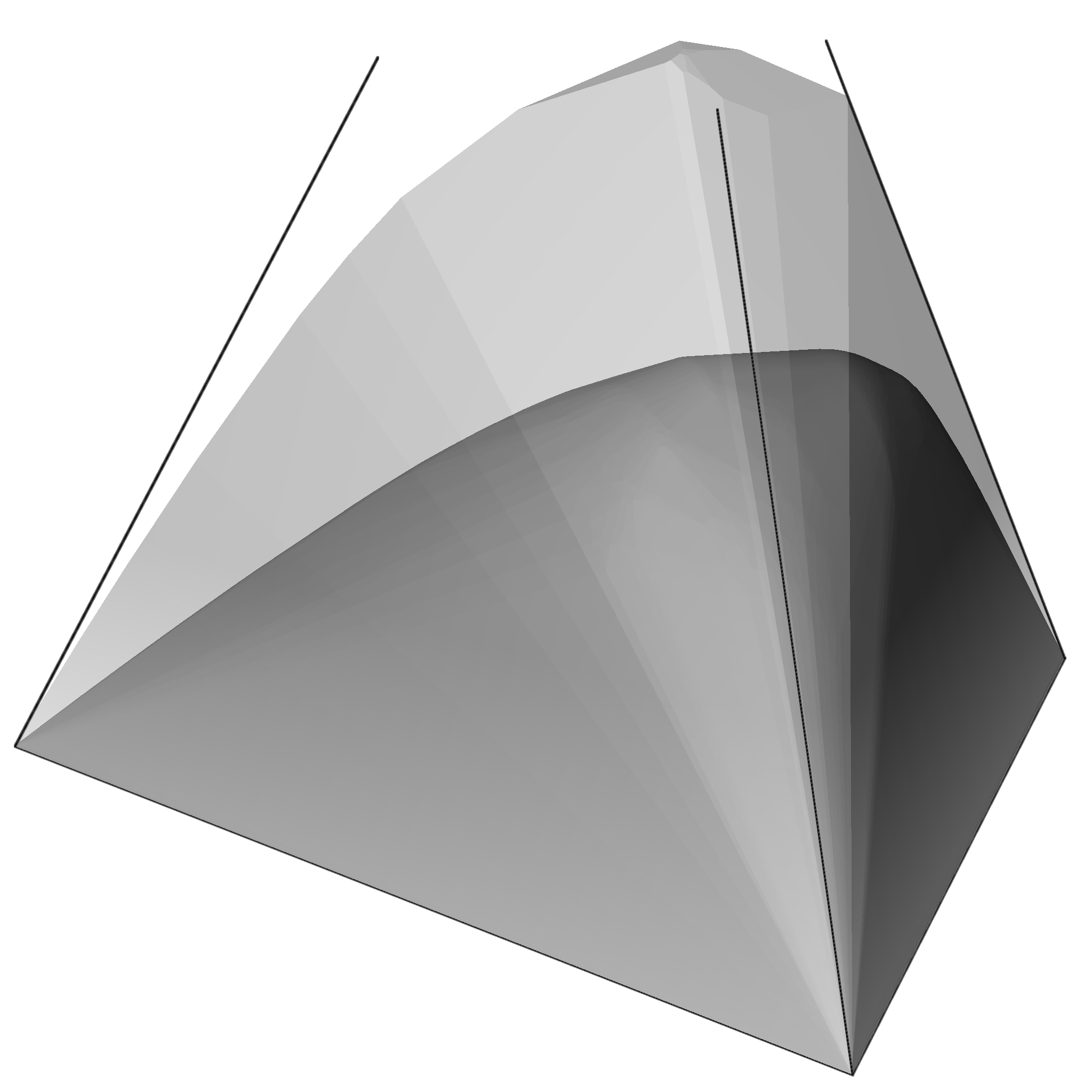}}}%
\setlength\unitlength{\ht\imagebox}%
\scalebox{0.0525}{\includegraphics{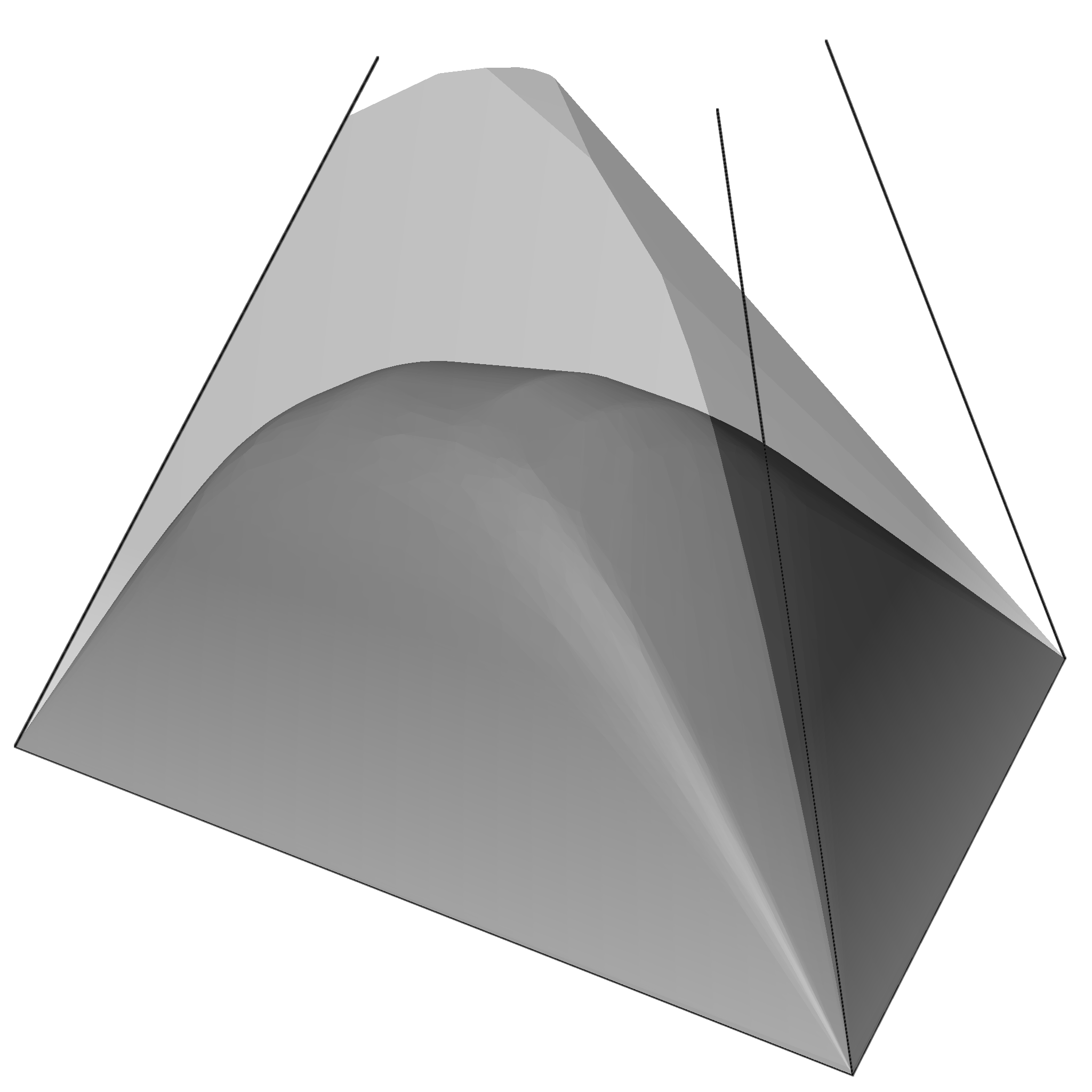}}\begin{picture}(0,0)(0,0)
\put(-1.0,0.205){\makebox{$\bde$}}
\put(-0.26,-0.09){\makebox{$\bbe$}}
\put(-0.04,0.27){\makebox{$\bga$}}
\put(-0.59,0.66){\circle{0.03}}
\put(-0.488,0.65){\circle*{0.03}}
\end{picture}%
\hfill
\usebox\imagebox\begin{picture}(0,0)(0,0)
\put(-1.02,0.212){\makebox{$\bbe$}}
\put(-0.26,-0.06){\makebox{$\bga$}}
\put(-0.04,0.255){\makebox{$\bde$}}
\put(-0.35,0.665){\circle*{0.03}}
\put(-0.23,0.66){\circle{0.03}}
\end{picture}%
\hfill
\scalebox{0.0525}{\includegraphics{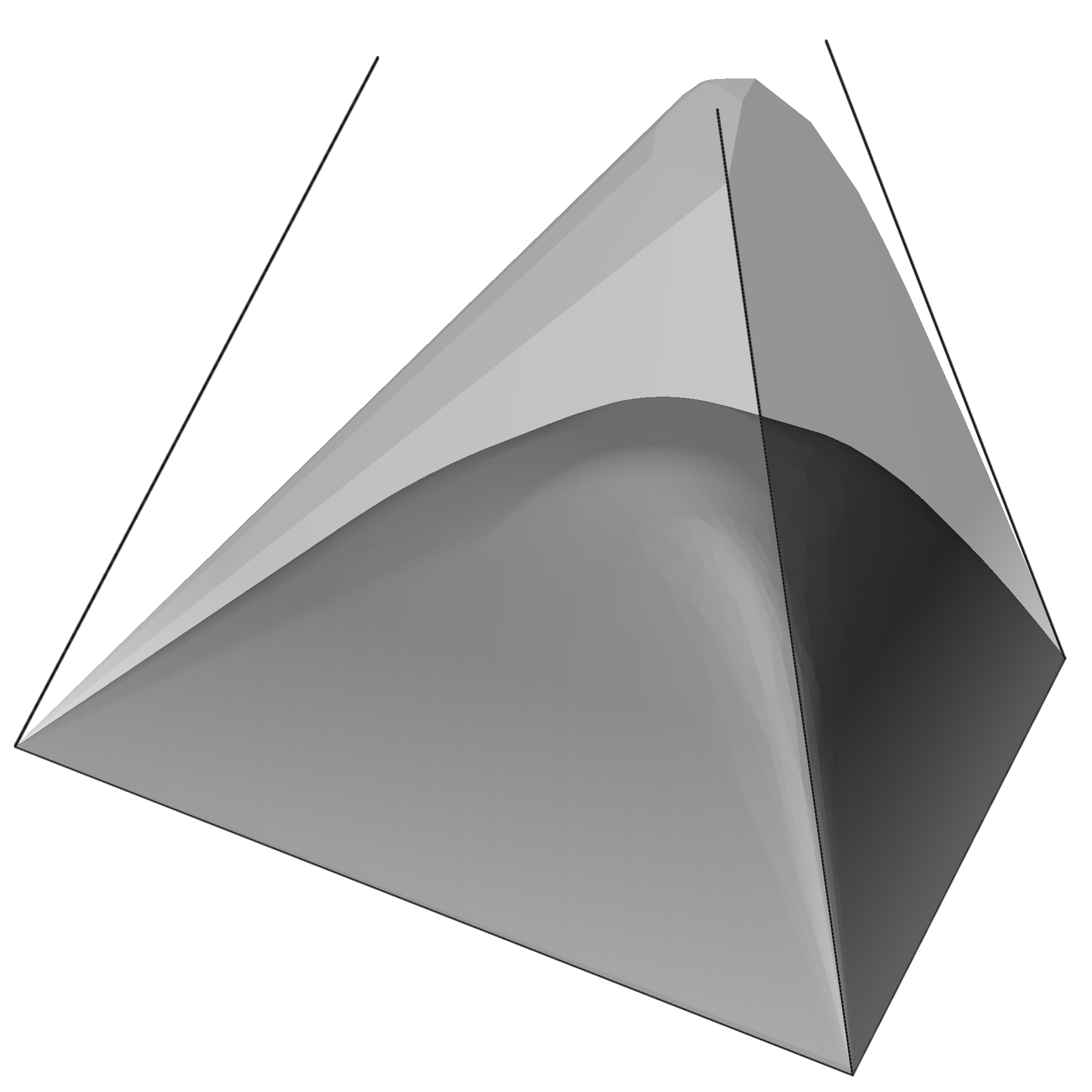}}%
\begin{picture}(0,0)(0,0)
\put(-1.0,0.232){\makebox{$\bga$}}
\put(-0.26,-0.09){\makebox{$\bde$}}
\put(-0.05,0.26){\makebox{$\bbe$}}
\put(-0.37,0.63){\circle{0.03}}
\put(-0.385,0.58){\circle*{0.03}}
\end{picture}%
\end{center}
\caption{Inner and outer approximations of $\Sec_{ij}$.}\label{fig:comp}
\end{figure}
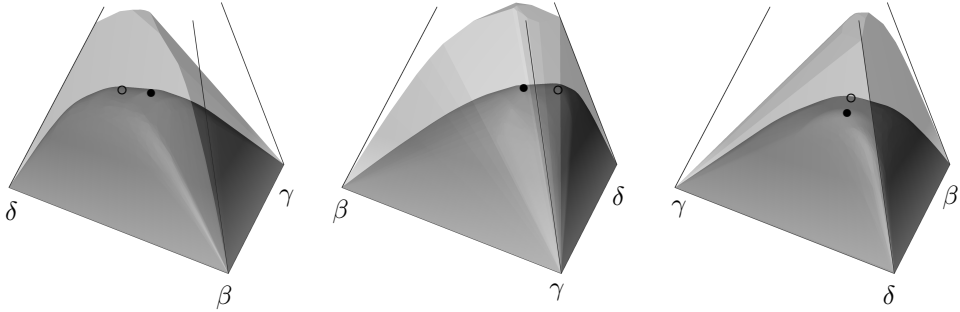

 Computer experiments were run to visualize $\Sec_{ij}$. Involved random
 variables were limited to take at most $11$ values. Various maximization
 procedures were run numerically over the distributions of four tuples
 of random variables. The corresponding entropy functions $f$ were
 transformed to $g=C_{ij} A_{i,j}B_{ij,k}\fti$ and then to
 \mbox{$h=g/g(N)$}, which is the convex combination
 \[
    h=\bar{\alpha}_h\bal+\bar{\beta_h}\bbe+\bar{\gamma}_h\bga
        +\bar{\delta}_h\bde \in \Sec_{ij}\,.
 \]
 The procedures maximized $\bar{\alpha}_h$ in various directions, over
 the distributions. Different methods and strategies were employed,
 including also randomized search. In this way, over $5$ million points
 from $\Sec_{ij}$ have been generated.

 In Figure~\ref{fig:comp}, the convex hull of these points is depicted
 as a dark gray region from three different perspectives. In the images, the
 vertex $\bal$ is missing and the straight lines are the incomplete edges of the
 tetrahedron incident to~$\bal$. The dark gray region is spanned by about $2200$
 extreme points. The projections of the extreme points from $\bal$ to $\bbe\bga\bde$
 do not exhaust the triangle uniformly, see Figure~\ref{fig:proextpoints}. This
 explains the lack of smoothness of the dark gray region. The two extreme points of
 the dark gray region depicted in Figure~\ref{fig:comp} are discussed in Section~\ref{S:4-atom}.

\begin{figure}[h!tb]
    \begin{center}\begin{tikzpicture}[scale=0.45]
                    \input{hpts.tex}
                  \end{tikzpicture}
    \end{center}
\caption{Extreme points of the dark gray region, projected to $\bbe\bga\bde$.}
    \label{fig:proextpoints}
\end{figure}
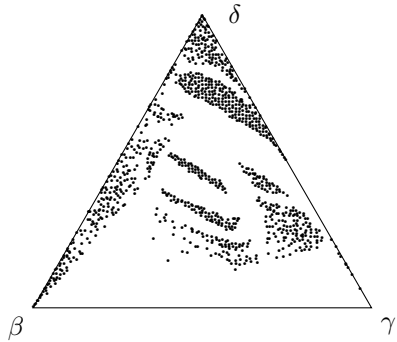

 The light gray region in Figure~\ref{fig:comp} visualizes an outer
 approximation of $\Sec_{ij}$ which was constructed from hundreds of known
 non-Shannon information inequalities, mostly from those of \cite{DFZ.nonSh}.
 Details are omitted. The gap between the approximations is large.

\begin{figure}[h!tb]
\begin{center}\hbox to \textwidth{\hss
\begin{tikzpicture}[scale=0.3]
\input{dbc.tex}
\draw[thick] (10,0) node[below]{$\gamma$}--(5,8.660254) node[right]{$\alpha$}--(0,0) node[below] {$\delta$};
\draw(0,0)--(10,0);
\end{tikzpicture}\rule{8mm}{0mm}%
\begin{tikzpicture}[scale=0.3]
\input{bcd.tex}
\draw[thick](10,0) node[below]{$\delta$}--(5,8.660254) node[right]{$\alpha$}--(0,0) node[below]{$\beta$};
\draw(0,0)--(10,0);
\end{tikzpicture}\rule{8mm}{0mm}%
\begin{tikzpicture}[scale=0.3]
\input{cdb.tex}
\draw[thick](10,0) node[below]{$\beta$}--(5,8.660254) node[right]{$\alpha$}--(0,0) node[below]{$\gamma$};
\draw(0,0)--(10,0);
\end{tikzpicture}%
\hss}\end{center}\vskip -20pt
\caption{Projections of the approximations of $\Sec_{ij}$ to triangles.}
\label{fig:vertex-view}
\end{figure}
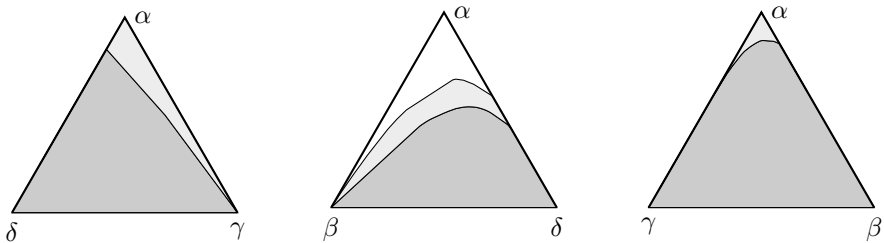

 Figure~\ref{fig:vertex-view} shows the dark and light gray regions
 when projected from the vertex $\bbe$/$\bga$/$\bde$ to the opposite triangle
 of the tetrahedron. By the non-Shannon inequalities \eqref{E:DFGsequence}
 discussed in the next section, the only almost entropic points
 on the edges $\bal\bbe$ and $\bal\bga$ are $\bbe$ and $\bga$. The
 analogous statement for the edge $\bal\bde$ is open.

\section{Ingleton score\label{S:4-atom}}

 As before, the ground set $N$ has four elements and $ij$ is
 a two-element subset of~$N$. The \emph{Ingleton score} of a polymatroidal
 rank function $h\neq0$ is defined as $\Iij(h)\triangleq\stv{h}{ij}/h(N)$
 \cite[Definition~3]{DFZ.nonSh}. The number
 \[
   \mathbb{I}^*\triangleq\inf\big\{\Iij(h)\colon 0\neq h\in\Hfent\big\}
 \]
 is referred to as the \emph{infimal} Ingleton score. This is likely
 the most interesting number related to the entropy region of four variables.
 By symmetry, $\mathbb{I}^*$ does not depend on~$ij$. This section presents
 an alternative way of minimization and a new upper bound on this number
 in Example~\ref{Ex:L}.

 First, the minimization is reduced to a three dimensional body.

\begin{theorem}\label{T:Istar}
  $\mathbb{I}^*=\min\nolimits_{\Sec_{ij}}\,\Iij\,$.
\end{theorem}

\begin{proof}
 Since the score is constant along rays and $\mathbb{I}^*$ is negative
 \[
   \mathbb{I}^*=\min\big\{\Iij(h)\colon h(N)=1\,,\,\stv{h}{ij}\leq0\;\;\text{and}
        \;\;h\in\clHfent\big\}\,,
 \]
 minimizing over a compact. If $h\in\HIvij$ then $\Iij(h)\geq\Iij(\hti)$ for $\hti\neq0$,
 and $\Iij(h)=0$ for $\hti=0\neq h$. Hence,
 \begin{equation}\label{E:lij}
   \mathbb{I}^*=\min\big\{\Iij(h)\colon h(N)=1\;\;\text{and}
        \;\;h\in\Lf{ij}\big\}\,.
 \end{equation}
 Recall that $\Lf{ij}=\clHfentti\cap\HIvij$ is the cone of tight almost entropic
 rank functions $h$ with $\stv{h}{ij}\leq0$. (By Corollary~\ref{C:4ri.ent},
 the above minimization can be expressed by special entropy functions.)

 If $g\in\Lf{ij}$ then \eqref{E:ABg} and tightness of $g$ provide
 \[
    A_{i,j}B_{ij,k}g(N)=g(N)-\trn{g}{kl|ij}=g(ij)\,.
 \]
 By Lemma~\ref{L:AB}, $\Iij(g)\geq\Iij(A_{i,j}B_{ij,k}g)$ when $g(ij)>0$. If $g(ij)=0$
 then $\stv{g}{ij}=g(k)+g(l)-g(kl)$ which is possible only if
 $\Iij(g)$ vanishes. Hence, Theorem~\ref{T:face} implies that the
 minimization restricts to $\Ff{ij}$. The assertion follows by symmetrization.
\end{proof}

 The three dimensional body $\Sec_{ij}$ is enclosed in the tetrahedron $\bal\bbe\bga\bde$
 and $-4\Iij(h)$ is the weight $\bar{\alpha}_h$ of $h\in\Sec_{ij}$ at the vertex
 $\bal$ when $h$ is written as the unique convex combination of the vertices.
 Thus, points of $\Sec_{ij}$ with the heaviest weight at $\bal$ are the minimizers
 in Theorem~\ref{T:Istar}. It should be also mentioned that is not clear
 which part of $\Sec_{ij}$ is exhausted by the very entropic points.

 Lower bounds on $\mathbb{I}^*$ can be obtained by relaxing $\Lf{ij}$ in \eqref{E:lij}.
 The simplest relaxation is to $\Hfti\cap\HIvij$ because this cone has only one
 extreme ray allowing for negative scores, namely the one generated by~$\bar{r}_{ij}$.
 Therefore, the infimal score $\mathbb{I}^*$ is lower bounded by
 $\Iij(\bar{r}_{ij})=-\frac14$. With little more work, the bound
 $-\frac16$ can be obtained by Zhang-Yeung inequality. Better lower
 bounds are reported in~\cite{DFZ.nonSh}, based on further non-Shannon-type
 inequalities.

 Upper bounds on the infimal Ingleton score arise from entropic
 polymatroids that violate the Ingleton inequality. There are many examples
 at disposal \cite{St.countex,M.4var.II,M.4var.III, ZhY.ineq,HRShVe,FShHa,
 HassSha,MaoThillHass,BostonNan,Walsh-Weber}. The following one has attracted
 a special attention.

\begin{example}\label{Ex:4at}
 Let $\xi_i$ and $\xi_j$ be exchangeable and $0$-$1$ valued, and $\xi_i=1$
 with the probability $\frac12$. Let further $\xi_k=\min\{\xi_i,\xi_j\}$ and
 $\xi_l=\max\{\xi_i,\xi_j\}$, see \cite[Example~1]{M.4var.II} or \cite{Csir}.
 If $0\leq p\leq\frac12$ denotes the probability of $\xi_i\xi_j=00$ and $h_p$
 is the entropy function of $\xi_i\xi_j\xi_k\xi_l$ then
 \[
  \Iij(h_p)=\frac{\trn{h_p}{ij|\pmn}-\trn{h_p}{kl|\pmn}}{h_p(N)}
 \]
 using the identity $\stv{}{ij}=\trn{}{kl|i}+\trn{}{kl|j}+\trn{}{ij|\pmn}-\trn{}{kl|\pmn}$.
 Let $\kappa(u)=-u\,\ln u$, $u>0$, and $\kappa(0)=0$.
 The numerator is
 \[\begin{split}
   2\ln 2-2\kappa(p)-2\kappa(\tfrac12-p)
   -\Big[2\kappa(1-p)-\kappa(1-2p)\Big]\\
   =(2p+1)\ln2-2\kappa(p)-2\kappa(1-p)
   \end{split}
 \]
 and the denominator is $2\kappa(p)+2\kappa(\frac12-p)$. The function
 $p\mapsto\Iij(h_{p})$ is strictly convex and has a unique global
 minimizer $p^*$ in the interval $[0,\tfrac12]$. Approximately,
 $p^*\doteq0.350457$ and $\Iij(h_{p^*})\doteq-0.089373$.
\end{example}

 The guess that $\mathbb{I}^*$ be equal to $\Iij(h_{p^*})$
 goes back to \cite{Csir} but the formulation \cite[Conjecture 4.1]{Csir}
 had a wrong numerical value. The same surmise appeared later in~\cite{DFZ.nonSh}
 as \emph{Four-atom conjecture}, referring to the four possible values of
 $\xi_i\xi_j\xi_k\xi_l$. The minimization was considered also in \cite{MaoThillHass,BostonNan}
 that report no score below $\Iij(h_{p^*})$. However, the computer experiments
 discussed in~Section~\ref{S:symmetrized} found an entropic point that
 can be transformed to an almost entropic point witnessing
 failure of Four-atom conjecture.

\begin{example}\label{Ex:L}
 Let each of four variables in $\xi_i\xi_j\xi_k\xi_l$ take values
 in $\{0,1,2,3\}$ and $p,q,r,s,t$ be nonnegative such that $p+q+r+s+t=\tfrac18$.
 The table below lists $40$ different configurations of the random vector.
 Each configuration in any column is attained with the probability given
 by the label of that column. The remaining configurations have zero
 probabilities. The corresponding entropy function is denoted by $f$.
    \begin{table}[htb]
    \begin{tabular}{|c|c|c|c|c|}\hline
        $p$&$q$&$r$&$s$&$t$\\\hline
        \strut 0000 & 0210 & 0011 & 0010 & 0001 \\
        \strut 0101 & 0321 & 0120 & 0121 & 0100 \\
        \strut 1010 & 1100 & 1002 & 1000 & 1012 \\
        \strut 1212 & 1332 & 1230 & 1232 & 1210 \\
        \strut 2121 & 2001 & 2103 & 2101 & 2123 \\
        \strut 2323 & 2233 & 2331 & 2333 & 2321 \\
        \strut 3232 & 3012 & 3213 & 3212 & 3233 \\
        \strut 3333 & 3123 & 3322 & 3323 & 3332 \\\hline
    \end{tabular}
    \end{table}

 By inspection of the table, in each column any variable takes each
 value twice.  Hence, $f(i)$, $f(j)$, $f(k)$ and $f(l)$ are equal to $2\ln2$.
 In each column, $\xi_i\xi_l$ and $\xi_j\xi_k$ are in the configurations
 $00, 33, 01, 10, 12, 21, 23, 32$. Hence, $f(il)$ and $f(jk)$ are equal
 to $3\ln2$. In each column but the second/third one, $\xi_i\xi_j$ and
 $\xi_k\xi_l$ are in the configurations $00, 33, 01, 10, 12, 21, 23, 32$,
 otherwise in $11, 22, 02, 20, 13, 31, 03, 30$. Hence,
 \[\begin{split}
    f(ij)&=8\kappa(q)+8\kappa(p+r+s+t) \\  
    f(kl)&=8\kappa(r)+8\kappa(p+q+s+t)\,.  
 \end{split}\]
 In the first and fifth/forth column, $\xi_i\xi_k$ and $\xi_j\xi_l$ are in the configurations
 $00$, $11$, $22$, $33$, each one attained twice, otherwise in $01,10,02,20,13,31,23,32$. Hence,
 \[\begin{split}
    f(ik)&=4\kappa(2p+2t)+8\kappa(q+r+s)\\     
    f(jl)&=4\kappa(2p+2s)+8\kappa(q+r+t)\,.    
 \end{split}\]
 Analogous considerations provide
 \[\begin{split}
     f(ikl)&=8\kappa(p+t)+8\kappa(q+s)+8\kappa(r)\\ 
     f(jkl)&=8\kappa(p+s)+8\kappa(q+t)+8\kappa(r)\\ 
     f(ijk)&=8\kappa(p+t)+8\kappa(r+s)+8\kappa(q)\\ 
     f(ijl)&=8\kappa(p+s)+8\kappa(r+t)+8\kappa(q)\,.
 \end{split}\]
 Since the $40$ configurations of the table are all different
 \[
     f(ijkl)=8\kappa(p)+8\kappa(q)+8\kappa(r)+8\kappa(s)+8\kappa(t)\,.  
 \]
 The choice
 \[
    \text{$p=0.09524$, $q=0.02494$, $r=0.00160$ and $s=t=0.00161$,}
 \]
 where $r$ is close to $s=t$, gives $\Iij(f)\doteq-0.078277$.
 This is yet bigger than the value $-0.089373$ from Example~\ref{Ex:4at}.
 However, $\Iij(\fti)\doteq-0.0912597$, refuting Four-atom conjecture.
 Even better, if $g$ denotes $A_{i,j}B_{ij,k}\fti$ then $\stv{g}{ij}=\stv{f}{ij}$
 by Lemma~\ref{L:AB}, and
 \[\begin{split}
    g(N)=\fti(N)-\trn{\fti}{kl|ij}
        &=2\fti(N)+\fti(ij)-\fti(ijk)-\fti(ijl)\\
        &=f(ij)+f(ikl)+f(jkl)-2f(N)<\fti(N)
 \end{split}\]
 by \eqref{E:ABg}. Hence, the score $\Iij(g)$ is approximately $-0.09243$,
 currently the best upper bound on the infimal Ingleton score.
\end{example}

 Figure~\ref{fig:comp} features also two extreme points of the dark gray region.
 The circle depicts the point $C_{ij}A_{i,j}h_{p^*}$ where $h_{p^*}$ was described in
 Example~\ref{Ex:4at}. The bullet depicts $C_{ij}A_{i,j}B_{ij,k}\fti$
 where $f$ is the entropic point from Example~\ref{Ex:L}.

\smallskip
 Figure~\ref{fig:face-views} shows the intersections of the dark and
 light gray regions, approximating $\Sec_{ij}$, with the triangles $\bal\bbe\bga$,
 $\bal\bga\bde$ and $\bal\bde\bbe$, two more exceptional points of $\Sec_{ij}$,
 and the role of Zhang-Yeung inequality.

 The symmetrized Zhang-Yeung inequality
 \[
     2\stv{h}{ij}+\big[\trn{h}{ik|l}+\trn{h}{il|k}+\trn{h}{kl|i}\big]
     +\big[\trn{h}{jk|l}+\trn{h}{jl|k}+\trn{h}{kl|j}\big]\geq0\,,
 \]
 valid for $h\in\Hfent$, rewrites to $\bar{\beta}_h+\bar{\delta}_h\geq\tfrac12\bar{\alpha}_h$.
 The plane defined by the equality here is indicated in Figure~\ref{fig:face-views}
 by the three dashed segments.

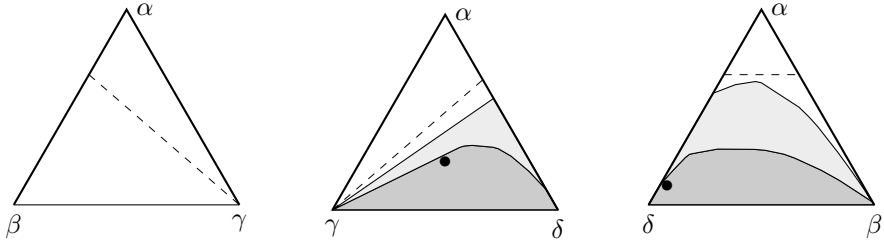
\begin{figure}[h!tb]
\begin{center}\hbox to\textwidth{\hss
\begin{tikzpicture}[scale=0.3]
\draw[thick](10,0) node[below]{$\gamma$}--(5,8.660254)node[right]{$\alpha$}--(0,0)node[below]{$\beta$};
\draw(0,0)--(10,0);
\draw[ultra thin,dashed](10/3,8.660254*2/3)--(10,0);
\end{tikzpicture}\rule{8mm}{0mm}%
\begin{tikzpicture}[scale=0.3]
\input{sidegd}
\draw[thick](10,0) node[below]{$\delta$}--(5,8.660254)node[right]{$\alpha$}--(0,0)node[below]{$\gamma$};
\draw(0,0)--(10,0);
\draw  (5.0000,2.1227) node{$\bullet$};
\draw[ultra thin,dashed](0,0)--(20/3,8.660254*2/3);
;
\end{tikzpicture}\rule{8mm}{0mm}%
\begin{tikzpicture}[scale=0.3]
\input{sidedb}
\draw[thick](10,0) node[below]{$\beta$}--(5,8.660254) node[right]{$\alpha$}--(0,0)node[below]{$\delta$};
\draw(0,0)--(10,0);
\draw (0.8244,0.8353) node{$\bullet$};
\draw[ultra thin,dashed](10/3,8.660254*2/3)--(20/3,8.660254*2/3);
\end{tikzpicture}%
\hss}\end{center}\vskip -20pt
\caption{Intersections of the approximations of $\Sec_{ij}$ with triangles.}\label{fig:face-views}
\end{figure}

 By \cite[Theorem 10]{DFZ.nonSh}, if $s\geq 0$ is integer then
 \[ \begin{split}
    (2^s-1)\,\stv{h}{ij} &+ \trn{h}{kl|i}
                    + s2^{s-1}\big[\trn{h}{ik|l}+\trn{h}{il|k} \big] \\
      &+ \big((s-2)2^{s-1}+1\big)\big[\trn{h}{jk|l}+\trn{h}{jl|k} \big]
        \geq 0  \,,\qquad h\in\Hfent\,.
    \end{split}
 \]
 This inequality and its instance with $i$ and $j$ interchanged sum to
 \begin{equation}\label{E:DFGsequence}
    \bar{\beta}_h +\big[(s-1)2^s+1\big]\bar{\delta}_h
    \geq \tfrac12(2^s-1)\bar{\alpha}_h\,,\qquad h\in\Hfent\,.
 \end{equation}
 Hence, the triangle $\bal\bbe\bga$ contains no almost
 entropic points except those on the edge $\bbe\bga$.

 The bullet inside the triangle $\bal\bga\bde$ depicts the entropy function
 $f_{\scriptscriptstyle 1/2}$ from Example~\ref{Ex:4at}, see also~\cite[Example~1]{M.4var.II}.
 The bullet inside the triangle $\bal\bbe\bde$ shows the almost entropic point
 $C_{ij}A_{i,j}\gti$ where $g$ is the entropy function discussed
 in Remark~\ref{R:sppoint}, see also~\cite[Example 2]{M.4var.II}.

\end{document}

%% file: hpts.tex
\draw (10,0)--(5,8.660254)--(0,0)--cycle;
\draw (-0.5,0) node[below]{$\beta$}
      (10.5,0) node[below]{$\gamma$}
      (5.5,8.7) node[right]{$\delta$};
\draw (0.016358,0.023794) node{.};
\draw (0.048247,0.083203) node{.};
\draw (0.086062,0.126166) node{.};
\draw (0.129237,0.218324) node{.};
\draw (0.152332,0.219418) node{.};
\draw (0.209184,0.288559) node{.};
\draw (0.231100,0.375217) node{.};
\draw (0.263183,0.397306) node{.};
\draw (0.313309,0.482461) node{.};
\draw (0.347215,0.553748) node{.};
\draw (0.416025,0.434429) node{.};
\draw (0.389372,0.501525) node{.};
\draw (0.389193,0.600638) node{.};
\draw (0.437413,0.686027) node{.};
\draw (0.520654,0.614083) node{.};
\draw (0.479596,0.713150) node{.};
\draw (0.508496,0.809801) node{.};
\draw (0.546737,0.896901) node{.};
\draw (0.632019,0.807872) node{.};
\draw (0.616824,0.858367) node{.};
\draw (0.608404,1.005759) node{.};
\draw (0.662972,0.725055) node{.};
\draw (0.682918,0.840698) node{.};
\draw (0.708836,0.892872) node{.};
\draw (0.692990,1.009662) node{.};
\draw (0.660345,1.137925) node{.};
\draw (0.702243,1.172470) node{.};
\draw (0.794305,0.872951) node{.};
\draw (0.811918,0.981449) node{.};
\draw (0.815731,1.116159) node{.};
\draw (0.789082,1.235555) node{.};
\draw (0.810275,1.309173) node{.};
\draw (0.921348,1.003611) node{.};
\draw (0.858199,1.109805) node{.};
\draw (0.918575,1.292747) node{.};
\draw (0.857098,1.436260) node{.};
\draw (0.913302,1.485645) node{.};
\draw (0.914504,1.551501) node{.};
\draw (1.007592,1.111871) node{.};
\draw (0.998137,1.192135) node{.};
\draw (1.011320,1.507525) node{.};
\draw (1.009729,1.550053) node{.};
\draw (1.018949,1.683045) node{.};
\draw (1.031599,1.754384) node{.};
\draw (1.111613,1.361829) node{.};
\draw (1.062828,1.467552) node{.};
\draw (1.089482,1.636413) node{.};
\draw (1.073348,1.673680) node{.};
\draw (1.126362,1.913196) node{.};
\draw (1.204366,1.495965) node{.};
\draw (1.238146,1.757991) node{.};
\draw (1.194584,1.915865) node{.};
\draw (1.183197,2.005156) node{.};
\draw (1.293385,1.429659) node{.};
\draw (1.292010,2.228313) node{.};
\draw (1.407726,1.596455) node{.};
\draw (1.380920,1.663396) node{.};
\draw (1.438890,1.912738) node{.};
\draw (1.443681,2.161680) node{.};
\draw (1.415321,2.291773) node{.};
\draw (1.457939,1.611023) node{.};
\draw (1.542878,2.279093) node{.};
\draw (1.613115,2.038049) node{.};
\draw (1.627454,2.594524) node{.};
\draw (1.643587,2.840617) node{.};
\draw (1.675925,1.941400) node{.};
\draw (1.716803,2.261610) node{.};
\draw (1.736969,2.473424) node{.};
\draw (1.687218,2.613060) node{.};
\draw (1.652733,2.840006) node{.};
\draw (1.690528,2.855157) node{.};
\draw (1.796752,2.493708) node{.};
\draw (1.768357,2.563355) node{.};
\draw (1.800513,2.768222) node{.};
\draw (1.829951,2.885361) node{.};
\draw (1.851374,2.146812) node{.};
\draw (1.908645,2.298554) node{.};
\draw (1.851283,2.389293) node{.};
\draw (1.889257,2.478875) node{.};
\draw (1.946501,2.603311) node{.};
\draw (1.945271,2.742306) node{.};
\draw (1.862977,2.801423) node{.};
\draw (1.949740,2.881449) node{.};
\draw (1.938914,3.131806) node{.};
\draw (2.043821,2.624380) node{.};
\draw (2.001864,2.677321) node{.};
\draw (1.972670,2.889119) node{.};
\draw (2.014910,3.055546) node{.};
\draw (2.029822,3.215746) node{.};
\draw (2.026804,3.407697) node{.};
\draw (2.032073,3.508034) node{.};
\draw (2.122729,2.405617) node{.};
\draw (2.079850,2.869500) node{.};
\draw (2.087114,2.989928) node{.};
\draw (2.088671,3.207886) node{.};
\draw (2.130728,3.391875) node{.};
\draw (2.071966,3.559426) node{.};
\draw (2.178613,2.686669) node{.};
\draw (2.202768,2.839580) node{.};
\draw (2.230870,3.000239) node{.};
\draw (2.220141,3.071548) node{.};
\draw (2.203289,3.225071) node{.};
\draw (2.178737,3.270325) node{.};
\draw (2.179868,3.491406) node{.};
\draw (2.155108,3.557135) node{.};
\draw (2.186252,3.737447) node{.};
\draw (2.245225,3.872937) node{.};
\draw (2.347370,2.809304) node{.};
\draw (2.340187,2.904941) node{.};
\draw (2.276488,3.203904) node{.};
\draw (2.310623,3.264704) node{.};
\draw (2.278085,3.404356) node{.};
\draw (2.313005,3.568982) node{.};
\draw (2.331121,3.748622) node{.};
\draw (2.307558,3.899711) node{.};
\draw (2.377848,2.761648) node{.};
\draw (2.372946,2.965091) node{.};
\draw (2.417616,3.207666) node{.};
\draw (2.404327,3.424049) node{.};
\draw (2.433744,3.556437) node{.};
\draw (2.415975,3.893068) node{.};
\draw (2.373231,4.012594) node{.};
\draw (2.459016,2.820894) node{.};
\draw (2.524877,2.917125) node{.};
\draw (2.505918,3.111663) node{.};
\draw (2.515705,3.307962) node{.};
\draw (2.469967,3.367029) node{.};
\draw (2.472724,3.487949) node{.};
\draw (2.495884,3.596598) node{.};
\draw (2.494857,3.771700) node{.};
\draw (2.500773,4.027127) node{.};
\draw (2.489796,4.083901) node{.};
\draw (2.499095,4.276351) node{.};
\draw (2.532669,4.365878) node{.};
\draw (2.578617,3.272507) node{.};
\draw (2.612795,3.433554) node{.};
\draw (2.582156,3.481523) node{.};
\draw (2.615144,3.641709) node{.};
\draw (2.595416,3.687987) node{.};
\draw (2.586138,3.851763) node{.};
\draw (2.619400,4.320507) node{.};
\draw (2.622669,4.521331) node{.};
\draw (2.670916,3.027183) node{.};
\draw (2.708082,3.051950) node{.};
\draw (2.675900,3.373736) node{.};
\draw (2.714934,3.769625) node{.};
\draw (2.713426,3.936746) node{.};
\draw (2.673975,3.990077) node{.};
\draw (2.682254,4.184966) node{.};
\draw (2.700578,4.523720) node{.};
\draw (2.762057,3.475810) node{.};
\draw (2.849720,3.579007) node{.};
\draw (2.787240,3.787278) node{.};
\draw (2.794994,3.969867) node{.};
\draw (2.824749,4.068269) node{.};
\draw (2.769223,4.199224) node{.};
\draw (2.754775,4.453601) node{.};
\draw (2.790491,4.701660) node{.};
\draw (2.832039,4.858243) node{.};
\draw (2.881417,3.617894) node{.};
\draw (2.919721,3.807616) node{.};
\draw (2.859890,4.060066) node{.};
\draw (2.879828,4.444966) node{.};
\draw (2.920212,4.763012) node{.};
\draw (2.926692,4.981850) node{.};
\draw (3.005410,3.367574) node{.};
\draw (3.011924,3.507407) node{.};
\draw (2.965543,3.696687) node{.};
\draw (3.018360,3.824692) node{.};
\draw (3.009356,4.095286) node{.};
\draw (2.969847,4.437766) node{.};
\draw (2.994963,4.687401) node{.};
\draw (3.092414,3.522030) node{.};
\draw (3.077330,4.214541) node{.};
\draw (3.065429,4.602210) node{.};
\draw (3.146755,4.730485) node{.};
\draw (3.106260,4.800582) node{.};
\draw (3.103215,5.343563) node{.};
\draw (3.170917,4.567515) node{.};
\draw (3.153227,4.931232) node{.};
\draw (3.206339,5.170254) node{.};
\draw (3.317768,3.823293) node{.};
\draw (3.266067,4.705692) node{.};
\draw (3.280292,4.903272) node{.};
\draw (3.332073,5.388518) node{.};
\draw (3.427706,3.891914) node{.};
\draw (3.439650,4.015639) node{.};
\draw (3.393933,4.113263) node{.};
\draw (3.389644,4.300660) node{.};
\draw (3.387070,4.400864) node{.};
\draw (3.443713,4.518353) node{.};
\draw (3.355786,4.731589) node{.};
\draw (3.383041,4.956631) node{.};
\draw (3.358696,5.426933) node{.};
\draw (3.366118,5.550057) node{.};
\draw (3.419127,5.853526) node{.};
\draw (3.444316,5.953923) node{.};
\draw (3.466673,4.126782) node{.};
\draw (3.452685,4.177507) node{.};
\draw (3.504138,4.287327) node{.};
\draw (3.539462,4.513993) node{.};
\draw (3.458823,4.581997) node{.};
\draw (3.502767,4.710435) node{.};
\draw (3.458890,5.836039) node{.};
\draw (3.476684,6.007948) node{.};
\draw (3.559570,2.051169) node{.};
\draw (3.639869,2.767890) node{.};
\draw (3.550469,4.221527) node{.};
\draw (3.564746,4.299523) node{.};
\draw (3.580001,4.466928) node{.};
\draw (3.553989,4.552968) node{.};
\draw (3.616225,4.883623) node{.};
\draw (3.580800,5.700471) node{.};
\draw (3.601684,5.784846) node{.};
\draw (3.608766,6.006100) node{.};
\draw (3.631762,6.166044) node{.};
\draw (3.671082,2.930637) node{.};
\draw (3.657577,4.484936) node{.};
\draw (3.680240,4.601129) node{.};
\draw (3.741152,4.832386) node{.};
\draw (3.715942,4.954332) node{.};
\draw (3.711653,5.515281) node{.};
\draw (3.695531,5.646068) node{.};
\draw (3.697347,5.667285) node{.};
\draw (3.705471,5.992295) node{.};
\draw (3.752160,2.546105) node{.};
\draw (3.848052,3.396012) node{.};
\draw (3.831935,3.486949) node{.};
\draw (3.766684,4.683397) node{.};
\draw (3.783139,4.810876) node{.};
\draw (3.784094,5.623876) node{.};
\draw (3.777841,5.743776) node{.};
\draw (3.801695,5.884019) node{.};
\draw (3.824039,6.322747) node{.};
\draw (3.916490,2.568708) node{.};
\draw (3.881500,2.667189) node{.};
\draw (3.932298,2.766627) node{.};
\draw (3.949683,3.345523) node{.};
\draw (3.905042,3.416612) node{.};
\draw (3.854011,4.778384) node{.};
\draw (3.915454,5.009832) node{.};
\draw (3.862390,5.556817) node{.};
\draw (3.850419,5.856625) node{.};
\draw (3.879127,6.607310) node{.};
\draw (3.946690,6.770792) node{.};
\draw (3.990056,1.382243) node{.};
\draw (3.957757,3.312958) node{.};
\draw (4.006878,3.385574) node{.};
\draw (4.023534,4.496460) node{.};
\draw (4.015947,4.954117) node{.};
\draw (4.028382,5.559951) node{.};
\draw (3.992320,5.723159) node{.};
\draw (4.003444,6.796810) node{.};
\draw (4.005127,6.899457) node{.};
\draw (4.127500,2.554342) node{.};
\draw (4.085136,3.177036) node{.};
\draw (4.077391,3.321039) node{.};
\draw (4.140565,4.438459) node{.};
\draw (4.104066,4.481899) node{.};
\draw (4.075489,4.557690) node{.};
\draw (4.096719,5.421921) node{.};
\draw (4.144591,5.646134) node{.};
\draw (4.111822,5.741964) node{.};
\draw (4.080839,6.943223) node{.};
\draw (4.072782,7.010033) node{.};
\draw (4.119685,7.104914) node{.};
\draw (4.153256,2.174705) node{.};
\draw (4.180300,2.384570) node{.};
\draw (4.232986,3.112795) node{.};
\draw (4.171699,3.203991) node{.};
\draw (4.176998,3.282782) node{.};
\draw (4.207309,4.386962) node{.};
\draw (4.200592,4.506512) node{.};
\draw (4.202011,5.640919) node{.};
\draw (4.160694,5.719814) node{.};
\draw (4.227826,6.664597) node{.};
\draw (4.230596,6.898352) node{.};
\draw (4.159638,7.039810) node{.};
\draw (4.213140,7.185666) node{.};
\draw (4.275951,3.139298) node{.};
\draw (4.288876,3.213772) node{.};
\draw (4.327521,4.247619) node{.};
\draw (4.324722,4.301280) node{.};
\draw (4.294051,4.409801) node{.};
\draw (4.260192,4.483173) node{.};
\draw (4.303144,5.594310) node{.};
\draw (4.260959,5.696193) node{.};
\draw (4.345108,6.638623) node{.};
\draw (4.322924,6.731077) node{.};
\draw (4.418554,1.657292) node{.};
\draw (4.406048,2.381788) node{.};
\draw (4.394183,3.139950) node{.};
\draw (4.408416,3.170675) node{.};
\draw (4.415812,4.201493) node{.};
\draw (4.401897,4.299695) node{.};
\draw (4.380142,4.379626) node{.};
\draw (4.413942,5.612132) node{.};
\draw (4.378170,5.721019) node{.};
\draw (4.438072,6.540573) node{.};
\draw (4.384577,6.599457) node{.};
\draw (4.429890,6.714382) node{.};
\draw (4.378990,6.790224) node{.};
\draw (4.357313,6.878044) node{.};
\draw (4.396754,7.595204) node{.};
\draw (4.477990,2.043085) node{.};
\draw (4.504460,2.416998) node{.};
\draw (4.505803,3.018222) node{.};
\draw (4.502691,3.110923) node{.};
\draw (4.509469,4.199966) node{.};
\draw (4.495153,4.295267) node{.};
\draw (4.471418,4.372099) node{.};
\draw (4.520231,6.404289) node{.};
\draw (4.488292,6.472241) node{.};
\draw (4.503300,6.604204) node{.};
\draw (4.499152,6.713878) node{.};
\draw (4.502281,6.800573) node{.};
\draw (4.498733,6.894672) node{.};
\draw (4.483502,7.471563) node{.};
\draw (4.520971,7.628953) node{.};
\draw (4.491515,7.650042) node{.};
\draw (4.546934,7.869415) node{.};
\draw (4.566194,1.644414) node{.};
\draw (4.648862,1.928675) node{.};
\draw (4.612832,2.990165) node{.};
\draw (4.588721,3.101330) node{.};
\draw (4.609731,4.039450) node{.};
\draw (4.613006,4.114564) node{.};
\draw (4.596085,4.195331) node{.};
\draw (4.598683,4.273173) node{.};
\draw (4.635688,6.331667) node{.};
\draw (4.592402,6.383011) node{.};
\draw (4.594213,6.500182) node{.};
\draw (4.624759,6.592200) node{.};
\draw (4.625904,6.658034) node{.};
\draw (4.644098,6.839727) node{.};
\draw (4.595613,6.883311) node{.};
\draw (4.628407,7.282192) node{.};
\draw (4.590741,7.354290) node{.};
\draw (4.606833,7.663213) node{.};
\draw (4.614372,7.910431) node{.};
\draw (4.632823,8.021272) node{.};
\draw (4.709409,2.246994) node{.};
\draw (4.715024,2.318508) node{.};
\draw (4.692614,2.354696) node{.};
\draw (4.696423,2.889523) node{.};
\draw (4.709526,3.016426) node{.};
\draw (4.694046,3.062212) node{.};
\draw (4.717053,4.023466) node{.};
\draw (4.707900,4.107489) node{.};
\draw (4.697773,4.194355) node{.};
\draw (4.670442,4.256288) node{.};
\draw (4.698308,6.395159) node{.};
\draw (4.666765,6.476056) node{.};
\draw (4.712190,6.585582) node{.};
\draw (4.670476,6.720086) node{.};
\draw (4.700555,6.783962) node{.};
\draw (4.730958,6.867485) node{.};
\draw (4.661320,6.962623) node{.};
\draw (4.725806,7.332777) node{.};
\draw (4.718710,7.412045) node{.};
\draw (4.692135,7.583725) node{.};
\draw (4.681857,7.652601) node{.};
\draw (4.728246,7.811146) node{.};
\draw (4.701630,7.909965) node{.};
\draw (4.707150,8.080684) node{.};
\draw (4.732736,8.179013) node{.};
\draw (4.810599,1.822464) node{.};
\draw (4.823699,2.828073) node{.};
\draw (4.770653,2.905677) node{.};
\draw (4.814791,2.993651) node{.};
\draw (4.799753,4.010540) node{.};
\draw (4.799947,4.084761) node{.};
\draw (4.798921,4.170351) node{.};
\draw (4.783526,6.230596) node{.};
\draw (4.806299,6.282162) node{.};
\draw (4.804288,6.382450) node{.};
\draw (4.819864,6.494592) node{.};
\draw (4.808644,6.638607) node{.};
\draw (4.791910,6.689604) node{.};
\draw (4.795646,6.802222) node{.};
\draw (4.833487,6.914896) node{.};
\draw (4.808329,7.208504) node{.};
\draw (4.828306,7.595728) node{.};
\draw (4.832953,7.793323) node{.};
\draw (4.788514,7.903766) node{.};
\draw (4.801269,7.993188) node{.};
\draw (4.775050,8.063643) node{.};
\draw (4.788291,8.223849) node{.};
\draw (4.841518,8.383273) node{.};
\draw (4.872748,1.805276) node{.};
\draw (4.852838,2.209783) node{.};
\draw (4.895373,2.256898) node{.};
\draw (4.935684,2.802208) node{.};
\draw (4.920253,2.918697) node{.};
\draw (4.874880,2.996918) node{.};
\draw (4.948485,3.839580) node{.};
\draw (4.912389,3.921247) node{.};
\draw (4.901946,4.010317) node{.};
\draw (4.890323,4.093498) node{.};
\draw (4.871020,4.150288) node{.};
\draw (4.888876,6.216531) node{.};
\draw (4.899184,6.307804) node{.};
\draw (4.889340,6.393573) node{.};
\draw (4.889745,6.506672) node{.};
\draw (4.897594,6.576240) node{.};
\draw (4.905680,6.700229) node{.};
\draw (4.898851,6.802559) node{.};
\draw (4.899581,6.912700) node{.};
\draw (4.903021,7.395181) node{.};
\draw (4.882579,7.488775) node{.};
\draw (4.912738,7.668660) node{.};
\draw (4.900403,7.779369) node{.};
\draw (4.902675,7.903729) node{.};
\draw (4.884767,7.979584) node{.};
\draw (4.948809,8.210742) node{.};
\draw (4.893114,8.317134) node{.};
\draw (4.911358,8.420347) node{.};
\draw (4.935063,8.520391) node{.};
\draw (5.024864,2.103351) node{.};
\draw (5.010591,2.190914) node{.};
\draw (4.965815,2.712639) node{.};
\draw (4.995348,2.818285) node{.};
\draw (5.002800,2.882985) node{.};
\draw (4.980220,2.957773) node{.};
\draw (4.979288,3.826181) node{.};
\draw (5.006660,3.908020) node{.};
\draw (4.998056,3.998156) node{.};
\draw (4.986902,4.069591) node{.};
\draw (5.002608,6.208878) node{.};
\draw (5.010208,6.331256) node{.};
\draw (4.987533,6.414009) node{.};
\draw (5.015307,6.493315) node{.};
\draw (5.002772,6.596994) node{.};
\draw (5.000334,6.705538) node{.};
\draw (5.002914,6.802125) node{.};
\draw (4.990615,6.894075) node{.};
\draw (5.025072,7.221778) node{.};
\draw (4.996238,7.440001) node{.};
\draw (4.996402,7.549898) node{.};
\draw (5.001160,7.586800) node{.};
\draw (4.983126,7.800149) node{.};
\draw (5.024330,7.926943) node{.};
\draw (5.004148,7.972604) node{.};
\draw (4.958656,8.053299) node{.};
\draw (5.026704,8.325478) node{.};
\draw (4.964417,8.374759) node{.};
\draw (4.999638,8.472892) node{.};
\draw (4.994422,8.616591) node{.};
\draw (5.000000,8.654213) node{.};
\draw (5.148331,1.737450) node{.};
\draw (5.139439,2.149022) node{.};
\draw (5.143515,2.610943) node{.};
\draw (5.104537,2.697520) node{.};
\draw (5.100007,2.805391) node{.};
\draw (5.055336,2.931705) node{.};
\draw (5.115338,3.810620) node{.};
\draw (5.109464,3.900200) node{.};
\draw (5.087225,3.975017) node{.};
\draw (5.117026,6.118214) node{.};
\draw (5.116905,6.207026) node{.};
\draw (5.111681,6.289837) node{.};
\draw (5.086148,6.386353) node{.};
\draw (5.111201,6.486983) node{.};
\draw (5.118898,6.603442) node{.};
\draw (5.107903,6.710053) node{.};
\draw (5.125059,6.788608) node{.};
\draw (5.112260,7.298808) node{.};
\draw (5.125157,7.489383) node{.};
\draw (5.076835,7.591253) node{.};
\draw (5.090328,7.680732) node{.};
\draw (5.090334,7.781180) node{.};
\draw (5.091815,7.890867) node{.};
\draw (5.050572,8.000707) node{.};
\draw (5.093270,8.126765) node{.};
\draw (5.108164,8.211074) node{.};
\draw (5.088268,8.279474) node{.};
\draw (5.142688,8.353143) node{.};
\draw (5.070677,8.518547) node{.};
\draw (5.058919,8.556194) node{.};
\draw (5.151650,1.977221) node{.};
\draw (5.246835,2.585627) node{.};
\draw (5.193651,2.682548) node{.};
\draw (5.198561,2.791830) node{.};
\draw (5.240465,3.703836) node{.};
\draw (5.205731,3.789249) node{.};
\draw (5.199348,3.889374) node{.};
\draw (5.239602,6.032243) node{.};
\draw (5.187426,6.097864) node{.};
\draw (5.215415,6.196390) node{.};
\draw (5.223527,6.306125) node{.};
\draw (5.215478,6.395240) node{.};
\draw (5.186977,6.501220) node{.};
\draw (5.198961,6.605554) node{.};
\draw (5.207744,6.684155) node{.};
\draw (5.205029,6.777170) node{.};
\draw (5.244228,7.219706) node{.};
\draw (5.199534,7.283712) node{.};
\draw (5.227260,7.384684) node{.};
\draw (5.215925,7.479494) node{.};
\draw (5.212207,7.589599) node{.};
\draw (5.180184,7.714475) node{.};
\draw (5.157062,7.933588) node{.};
\draw (5.203714,8.015097) node{.};
\draw (5.193729,8.144748) node{.};
\draw (5.184656,8.326215) node{.};
\draw (5.266464,1.541539) node{.};
\draw (5.292224,1.941985) node{.};
\draw (5.300659,1.998332) node{.};
\draw (5.282208,2.119194) node{.};
\draw (5.305272,2.548050) node{.};
\draw (5.284354,2.617024) node{.};
\draw (5.278086,2.699589) node{.};
\draw (5.323993,2.815105) node{.};
\draw (5.313180,3.699918) node{.};
\draw (5.299976,3.799769) node{.};
\draw (5.320074,6.010737) node{.};
\draw (5.296132,6.097647) node{.};
\draw (5.307991,6.202471) node{.};
\draw (5.317537,6.311019) node{.};
\draw (5.299087,6.403915) node{.};
\draw (5.295128,6.506889) node{.};
\draw (5.281926,6.611402) node{.};
\draw (5.285909,6.694644) node{.};
\draw (5.272645,6.755082) node{.};
\draw (5.287608,7.237892) node{.};
\draw (5.293012,7.327674) node{.};
\draw (5.312573,7.405087) node{.};
\draw (5.316951,7.483543) node{.};
\draw (5.283003,7.590134) node{.};
\draw (5.303201,7.685451) node{.};
\draw (5.282016,7.778073) node{.};
\draw (5.335904,7.892011) node{.};
\draw (5.271036,8.030502) node{.};
\draw (5.320497,8.058089) node{.};
\draw (5.397235,1.595608) node{.};
\draw (5.428904,2.015007) node{.};
\draw (5.424227,2.090026) node{.};
\draw (5.422341,2.510595) node{.};
\draw (5.385027,2.593953) node{.};
\draw (5.404524,2.718847) node{.};
\draw (5.442384,2.776611) node{.};
\draw (5.411548,3.601116) node{.};
\draw (5.410707,3.702105) node{.};
\draw (5.378431,3.776403) node{.};
\draw (5.420327,6.007469) node{.};
\draw (5.386311,6.121052) node{.};
\draw (5.392038,6.189229) node{.};
\draw (5.383360,6.328916) node{.};
\draw (5.374884,6.419329) node{.};
\draw (5.403717,6.518993) node{.};
\draw (5.420105,6.611752) node{.};
\draw (5.400617,6.695556) node{.};
\draw (5.390050,7.166551) node{.};
\draw (5.416690,7.252021) node{.};
\draw (5.398292,7.363901) node{.};
\draw (5.413461,7.476914) node{.};
\draw (5.356501,7.635852) node{.};
\draw (5.401634,7.707667) node{.};
\draw (5.390467,7.844947) node{.};
\draw (5.415812,7.913228) node{.};
\draw (5.511415,1.569037) node{.};
\draw (5.514186,1.863442) node{.};
\draw (5.505849,1.992375) node{.};
\draw (5.479209,2.448789) node{.};
\draw (5.495999,2.518983) node{.};
\draw (5.497513,2.573180) node{.};
\draw (5.470819,2.699782) node{.};
\draw (5.467749,3.539880) node{.};
\draw (5.508864,3.606604) node{.};
\draw (5.529597,3.686469) node{.};
\draw (5.514694,3.753269) node{.};
\draw (5.530514,5.921232) node{.};
\draw (5.501929,6.004503) node{.};
\draw (5.493209,6.095031) node{.};
\draw (5.485843,6.197120) node{.};
\draw (5.497448,6.294928) node{.};
\draw (5.525387,6.435536) node{.};
\draw (5.524197,6.492551) node{.};
\draw (5.492088,6.591756) node{.};
\draw (5.504335,6.691353) node{.};
\draw (5.502584,7.204569) node{.};
\draw (5.536980,7.315697) node{.};
\draw (5.507251,7.403004) node{.};
\draw (5.483557,7.487677) node{.};
\draw (5.575821,1.518390) node{.};
\draw (5.626367,1.660942) node{.};
\draw (5.595538,1.819742) node{.};
\draw (5.570224,1.863247) node{.};
\draw (5.598609,2.009520) node{.};
\draw (5.580562,2.424399) node{.};
\draw (5.586776,2.495451) node{.};
\draw (5.600878,2.598733) node{.};
\draw (5.558472,2.661830) node{.};
\draw (5.610290,3.497851) node{.};
\draw (5.606757,3.580645) node{.};
\draw (5.559519,3.659109) node{.};
\draw (5.594001,5.895492) node{.};
\draw (5.603169,6.005780) node{.};
\draw (5.601152,6.097601) node{.};
\draw (5.610530,6.183782) node{.};
\draw (5.615557,6.279519) node{.};
\draw (5.600656,6.398638) node{.};
\draw (5.618246,6.482816) node{.};
\draw (5.587515,6.606167) node{.};
\draw (5.600643,7.026212) node{.};
\draw (5.604670,7.097029) node{.};
\draw (5.609201,7.220064) node{.};
\draw (5.573024,7.270552) node{.};
\draw (5.586295,7.487920) node{.};
\draw (5.741603,1.420293) node{.};
\draw (5.732545,1.619338) node{.};
\draw (5.691525,1.798018) node{.};
\draw (5.682838,1.874426) node{.};
\draw (5.701378,1.952937) node{.};
\draw (5.712852,2.309666) node{.};
\draw (5.699624,2.406533) node{.};
\draw (5.714541,2.505332) node{.};
\draw (5.690578,3.403289) node{.};
\draw (5.662797,3.510381) node{.};
\draw (5.697722,5.810314) node{.};
\draw (5.685391,5.901194) node{.};
\draw (5.718705,5.982719) node{.};
\draw (5.695863,6.091413) node{.};
\draw (5.698986,6.208527) node{.};
\draw (5.684244,6.309847) node{.};
\draw (5.712114,6.369969) node{.};
\draw (5.693101,6.502414) node{.};
\draw (5.655005,6.554266) node{.};
\draw (5.736846,6.905967) node{.};
\draw (5.730278,7.006577) node{.};
\draw (5.684290,7.180564) node{.};
\draw (5.687090,7.311610) node{.};
\draw (5.711300,7.380280) node{.};
\draw (5.800551,1.410688) node{.};
\draw (5.783227,1.578653) node{.};
\draw (5.804929,1.744689) node{.};
\draw (5.795940,1.803889) node{.};
\draw (5.794404,1.874887) node{.};
\draw (5.758729,1.976158) node{.};
\draw (5.828765,2.328450) node{.};
\draw (5.792124,2.405112) node{.};
\draw (5.805346,2.491631) node{.};
\draw (5.761272,2.612271) node{.};
\draw (5.818193,5.665290) node{.};
\draw (5.813042,5.814700) node{.};
\draw (5.804003,5.903885) node{.};
\draw (5.798504,6.003030) node{.};
\draw (5.796136,6.101932) node{.};
\draw (5.790207,6.204423) node{.};
\draw (5.780669,6.281116) node{.};
\draw (5.812193,6.397444) node{.};
\draw (5.814590,6.491262) node{.};
\draw (5.848580,6.809452) node{.};
\draw (5.820047,6.919315) node{.};
\draw (5.818205,6.985327) node{.};
\draw (5.829950,7.053509) node{.};
\draw (5.755953,7.248475) node{.};
\draw (5.861288,1.318198) node{.};
\draw (5.914187,1.815642) node{.};
\draw (5.879393,1.902055) node{.};
\draw (5.896712,2.317310) node{.};
\draw (5.882992,2.408001) node{.};
\draw (5.935733,2.482719) node{.};
\draw (5.919184,2.626230) node{.};
\draw (5.892817,5.620684) node{.};
\draw (5.900351,5.692958) node{.};
\draw (5.883587,5.807551) node{.};
\draw (5.897651,5.918094) node{.};
\draw (5.897573,5.986422) node{.};
\draw (5.898396,6.106812) node{.};
\draw (5.910414,6.208601) node{.};
\draw (5.912850,6.306528) node{.};
\draw (5.884000,6.371883) node{.};
\draw (5.945639,6.705121) node{.};
\draw (5.895390,6.835313) node{.};
\draw (5.875765,6.896934) node{.};
\draw (5.906844,7.003794) node{.};
\draw (5.865768,7.127935) node{.};
\draw (5.984481,1.130247) node{.};
\draw (6.002597,1.299735) node{.};
\draw (5.988821,1.382647) node{.};
\draw (5.988257,1.730562) node{.};
\draw (6.006870,1.778690) node{.};
\draw (5.988772,1.888603) node{.};
\draw (5.966337,2.378660) node{.};
\draw (5.992336,2.483427) node{.};
\draw (6.044793,2.564797) node{.};
\draw (5.962694,2.651740) node{.};
\draw (6.025761,5.544219) node{.};
\draw (5.993501,5.620827) node{.};
\draw (5.979498,5.709582) node{.};
\draw (6.002740,5.806072) node{.};
\draw (5.997719,5.894047) node{.};
\draw (5.998609,6.002815) node{.};
\draw (6.021729,6.130451) node{.};
\draw (6.004802,6.204804) node{.};
\draw (5.959197,6.372127) node{.};
\draw (6.011826,6.719424) node{.};
\draw (5.988535,6.788122) node{.};
\draw (6.113816,1.263737) node{.};
\draw (6.126885,1.431806) node{.};
\draw (6.087752,1.540709) node{.};
\draw (6.097779,1.810651) node{.};
\draw (6.110387,1.876324) node{.};
\draw (6.075672,2.640069) node{.};
\draw (6.061974,2.652494) node{.};
\draw (6.109761,4.249212) node{.};
\draw (6.098765,5.524413) node{.};
\draw (6.104582,5.597142) node{.};
\draw (6.088304,5.724176) node{.};
\draw (6.099884,5.811888) node{.};
\draw (6.113430,5.880445) node{.};
\draw (6.100545,6.001010) node{.};
\draw (6.112011,6.103008) node{.};
\draw (6.134352,6.158691) node{.};
\draw (6.099350,6.609407) node{.};
\draw (6.184807,1.687116) node{.};
\draw (6.187838,1.837597) node{.};
\draw (6.220166,1.897776) node{.};
\draw (6.238561,4.021255) node{.};
\draw (6.189427,4.187081) node{.};
\draw (6.192627,5.498325) node{.};
\draw (6.208766,5.591808) node{.};
\draw (6.195061,5.717079) node{.};
\draw (6.201276,5.818439) node{.};
\draw (6.182234,5.937488) node{.};
\draw (6.206100,6.016552) node{.};
\draw (6.178602,6.090173) node{.};
\draw (6.172787,6.154763) node{.};
\draw (6.163879,6.485799) node{.};
\draw (6.260916,1.303104) node{.};
\draw (6.303039,1.397216) node{.};
\draw (6.319556,1.505293) node{.};
\draw (6.337335,1.803744) node{.};
\draw (6.310069,3.992492) node{.};
\draw (6.314593,4.109665) node{.};
\draw (6.306137,5.409718) node{.};
\draw (6.293217,5.493804) node{.};
\draw (6.301498,5.586435) node{.};
\draw (6.302694,5.689614) node{.};
\draw (6.311804,5.802789) node{.};
\draw (6.303382,5.897907) node{.};
\draw (6.298915,5.976306) node{.};
\draw (6.264385,6.382252) node{.};
\draw (6.254385,6.467768) node{.};
\draw (6.439280,1.372344) node{.};
\draw (6.415845,1.458222) node{.};
\draw (6.352744,1.886524) node{.};
\draw (6.402608,1.988895) node{.};
\draw (6.387767,3.924412) node{.};
\draw (6.409418,3.975457) node{.};
\draw (6.367766,4.119962) node{.};
\draw (6.410552,5.343753) node{.};
\draw (6.406204,5.417404) node{.};
\draw (6.411187,5.486569) node{.};
\draw (6.400695,5.581434) node{.};
\draw (6.400221,5.692732) node{.};
\draw (6.374164,5.797478) node{.};
\draw (6.418005,5.915224) node{.};
\draw (6.406560,5.986620) node{.};
\draw (6.403557,6.204039) node{.};
\draw (6.489577,3.699789) node{.};
\draw (6.519798,3.815495) node{.};
\draw (6.548991,3.943351) node{.};
\draw (6.507509,3.971063) node{.};
\draw (6.479902,4.088260) node{.};
\draw (6.501383,5.312801) node{.};
\draw (6.491661,5.375116) node{.};
\draw (6.480184,5.494921) node{.};
\draw (6.532382,5.592188) node{.};
\draw (6.494453,5.696068) node{.};
\draw (6.487888,5.777813) node{.};
\draw (6.506135,6.014747) node{.};
\draw (6.633200,1.330783) node{.};
\draw (6.591505,1.444363) node{.};
\draw (6.606138,1.534530) node{.};
\draw (6.626115,3.163820) node{.};
\draw (6.571357,3.619682) node{.};
\draw (6.627814,3.669526) node{.};
\draw (6.589164,3.812778) node{.};
\draw (6.630870,3.875121) node{.};
\draw (6.562719,3.981040) node{.};
\draw (6.618713,5.239561) node{.};
\draw (6.603057,5.295119) node{.};
\draw (6.596220,5.393386) node{.};
\draw (6.602279,5.515424) node{.};
\draw (6.597784,5.588940) node{.};
\draw (6.567156,5.737301) node{.};
\draw (6.578761,5.767972) node{.};
\draw (6.569447,5.924368) node{.};
\draw (6.743570,3.017910) node{.};
\draw (6.680711,3.156951) node{.};
\draw (6.684445,3.270992) node{.};
\draw (6.666495,3.578144) node{.};
\draw (6.689566,3.680181) node{.};
\draw (6.691660,3.846492) node{.};
\draw (6.702834,3.884106) node{.};
\draw (6.728717,5.124844) node{.};
\draw (6.697645,5.206657) node{.};
\draw (6.698430,5.335481) node{.};
\draw (6.694390,5.415649) node{.};
\draw (6.680711,5.517234) node{.};
\draw (6.702276,5.599133) node{.};
\draw (6.706590,5.682066) node{.};
\draw (6.786713,3.080429) node{.};
\draw (6.773620,3.523764) node{.};
\draw (6.776097,3.613107) node{.};
\draw (6.764414,3.728223) node{.};
\draw (6.817676,3.826803) node{.};
\draw (6.817531,3.877208) node{.};
\draw (6.814480,5.105112) node{.};
\draw (6.796939,5.209904) node{.};
\draw (6.808330,5.297488) node{.};
\draw (6.799654,5.396907) node{.};
\draw (6.921662,1.935539) node{.};
\draw (6.934724,2.036082) node{.};
\draw (6.919291,2.406782) node{.};
\draw (6.878530,2.562027) node{.};
\draw (6.888703,2.708063) node{.};
\draw (6.938721,2.892875) node{.};
\draw (6.897240,2.992266) node{.};
\draw (6.868865,3.117410) node{.};
\draw (6.931021,3.443063) node{.};
\draw (6.853695,3.508704) node{.};
\draw (6.903716,3.612300) node{.};
\draw (6.908779,3.692864) node{.};
\draw (6.895652,3.797022) node{.};
\draw (6.918597,5.034365) node{.};
\draw (6.905816,5.088301) node{.};
\draw (6.900173,5.198184) node{.};
\draw (6.882116,5.332985) node{.};
\draw (6.881074,5.357853) node{.};
\draw (6.956374,1.597902) node{.};
\draw (6.989174,2.060104) node{.};
\draw (6.968213,2.219672) node{.};
\draw (6.996840,2.396300) node{.};
\draw (7.018394,2.684747) node{.};
\draw (7.018320,2.905212) node{.};
\draw (7.007994,2.967016) node{.};
\draw (7.037013,3.065899) node{.};
\draw (7.009818,3.515660) node{.};
\draw (6.998226,3.603853) node{.};
\draw (6.985193,3.679235) node{.};
\draw (7.031398,4.935133) node{.};
\draw (7.013554,5.002900) node{.};
\draw (6.987814,5.113257) node{.};
\draw (6.977535,5.157822) node{.};
\draw (7.127730,2.206257) node{.};
\draw (7.088838,2.409806) node{.};
\draw (7.065400,2.606515) node{.};
\draw (7.061206,2.745944) node{.};
\draw (7.111478,2.831638) node{.};
\draw (7.120212,2.880638) node{.};
\draw (7.107903,3.011784) node{.};
\draw (7.092610,3.421019) node{.};
\draw (7.114530,3.614813) node{.};
\draw (7.072162,3.686366) node{.};
\draw (7.130966,4.831147) node{.};
\draw (7.106490,4.899490) node{.};
\draw (7.088922,4.998018) node{.};
\draw (7.204608,1.487903) node{.};
\draw (7.171971,2.093486) node{.};
\draw (7.246658,2.319239) node{.};
\draw (7.185110,2.482530) node{.};
\draw (7.197556,2.771128) node{.};
\draw (7.211283,2.953992) node{.};
\draw (7.201504,3.052040) node{.};
\draw (7.195617,3.396674) node{.};
\draw (7.199222,3.506965) node{.};
\draw (7.218238,3.579201) node{.};
\draw (7.186297,4.800880) node{.};
\draw (7.268558,1.681058) node{.};
\draw (7.251993,1.945955) node{.};
\draw (7.300576,2.355394) node{.};
\draw (7.268295,2.579159) node{.};
\draw (7.254191,2.783634) node{.};
\draw (7.273843,2.990311) node{.};
\draw (7.290884,3.429259) node{.};
\draw (7.278153,3.502752) node{.};
\draw (7.314103,4.615559) node{.};
\draw (7.285552,4.665962) node{.};
\draw (7.401352,1.505504) node{.};
\draw (7.393166,1.758326) node{.};
\draw (7.418641,1.984381) node{.};
\draw (7.423175,2.062040) node{.};
\draw (7.363095,2.229725) node{.};
\draw (7.402370,2.317329) node{.};
\draw (7.417356,2.428508) node{.};
\draw (7.378806,2.569151) node{.};
\draw (7.421805,2.712959) node{.};
\draw (7.385078,2.792596) node{.};
\draw (7.401650,3.411410) node{.};
\draw (7.431203,4.428294) node{.};
\draw (7.365059,4.536669) node{.};
\draw (7.468520,2.086135) node{.};
\draw (7.528330,2.189711) node{.};
\draw (7.502202,2.287754) node{.};
\draw (7.509563,2.407635) node{.};
\draw (7.461211,2.535310) node{.};
\draw (7.481983,2.818120) node{.};
\draw (7.486975,2.976823) node{.};
\draw (7.518727,3.325794) node{.};
\draw (7.467342,4.367548) node{.};
\draw (7.600102,1.573400) node{.};
\draw (7.564369,1.838916) node{.};
\draw (7.633710,1.894989) node{.};
\draw (7.620422,2.284105) node{.};
\draw (7.603394,2.410399) node{.};
\draw (7.574509,2.714888) node{.};
\draw (7.609564,2.798222) node{.};
\draw (7.614001,2.869323) node{.};
\draw (7.733100,1.943179) node{.};
\draw (7.675070,2.027815) node{.};
\draw (7.682664,2.201189) node{.};
\draw (7.748627,2.323519) node{.};
\draw (7.704964,2.436916) node{.};
\draw (7.715820,2.576753) node{.};
\draw (7.691430,2.722560) node{.};
\draw (7.729504,2.752441) node{.};
\draw (7.686044,2.863607) node{.};
\draw (7.694797,2.951702) node{.};
\draw (7.770115,1.681478) node{.};
\draw (7.813094,2.034514) node{.};
\draw (7.779037,2.132351) node{.};
\draw (7.788472,2.183523) node{.};
\draw (7.790019,2.279751) node{.};
\draw (7.808368,2.417799) node{.};
\draw (7.778676,2.485760) node{.};
\draw (7.812995,2.591107) node{.};
\draw (7.809408,2.718831) node{.};
\draw (7.799731,2.781452) node{.};
\draw (7.768578,2.881365) node{.};
\draw (7.938524,1.692110) node{.};
\draw (7.887933,1.781970) node{.};
\draw (7.882818,1.973152) node{.};
\draw (7.932450,2.140845) node{.};
\draw (7.912169,2.353807) node{.};
\draw (7.890376,2.678232) node{.};
\draw (7.886332,2.818356) node{.};
\draw (8.025322,1.718086) node{.};
\draw (8.036848,1.907576) node{.};
\draw (7.973781,2.011378) node{.};
\draw (7.971445,2.093255) node{.};
\draw (7.966873,2.176270) node{.};
\draw (7.988315,2.463926) node{.};
\draw (8.045897,2.633610) node{.};
\draw (8.003704,2.723132) node{.};
\draw (7.995817,2.794978) node{.};
\draw (8.077270,1.815313) node{.};
\draw (8.127874,2.000727) node{.};
\draw (8.057139,2.137713) node{.};
\draw (8.051126,2.226363) node{.};
\draw (8.072751,2.283747) node{.};
\draw (8.127509,2.411494) node{.};
\draw (8.064565,2.515231) node{.};
\draw (8.115288,2.625338) node{.};
\draw (8.166016,1.707957) node{.};
\draw (8.220278,1.816491) node{.};
\draw (8.200771,1.920684) node{.};
\draw (8.211410,2.094796) node{.};
\draw (8.216117,2.209346) node{.};
\draw (8.214945,2.322556) node{.};
\draw (8.180909,2.403954) node{.};
\draw (8.183030,2.475267) node{.};
\draw (8.320202,2.094568) node{.};
\draw (8.390832,1.791388) node{.};
\draw (8.397615,1.940688) node{.};
\draw (8.397750,1.983865) node{.};
\draw (8.405797,2.093226) node{.};
\draw (8.382037,2.233506) node{.};
\draw (8.366790,2.320478) node{.};
\draw (8.523813,1.982842) node{.};
\draw (8.537860,2.068706) node{.};
\draw (8.461627,2.289043) node{.};
\draw (8.820167,2.002797) node{.};
\draw (9.038559,1.598139) node{.};
\draw (8.959711,1.797843) node{.};
\draw (9.094761,1.512650) node{.};
\draw (9.236336,1.292928) node{.};
\draw (9.311987,1.187198) node{.};
\draw (9.252727,1.275261) node{.};
\draw (9.432277,0.977489) node{.};
\draw (9.659367,0.578137) node{.};

%% file: dbc.tex
\filldraw[fill=black!20]
(0.000000,0.000000)--
(1.225094,2.121924)--
(4.186122,7.250576)--
(6.785430,4.328257)--
(10.000000,0.000000)--
cycle;
\filldraw[fill=black!7]
(5.000000,8.660254)--
(1.225094,2.121924)--
(4.186122,7.250576)--
(6.785430,4.328257)--
(10.000000,0.000000)--
cycle;

%% file: bcd.tex
\filldraw[fill=black!20]
(0.000000,0.000000)--
(0.013209,0.017785)--
(3.651627,3.366333)--
(3.929186,3.617911)--
(4.056146,3.705726)--
(4.205725,3.804987)--
(4.241652,3.826618)--
(4.288111,3.854181)--
(4.325734,3.876462)--
(4.341661,3.885236)--
(4.416739,3.924861)--
(4.537071,3.981879)--
(4.845931,4.122283)--
(5.009647,4.195642)--
(5.201251,4.274486)--
(5.303203,4.311319)--
(5.530592,4.392227)--
(5.705181,4.426388)--
(5.803504,4.445628)--
(5.849555,4.453645)--
(5.972342,4.465300)--
(6.002238,4.467864)--
(6.065618,4.471182)--
(6.176604,4.469357)--
(6.280715,4.463595)--
(6.451473,4.443636)--
(6.561077,4.417136)--
(6.595115,4.407274)--
(6.686541,4.379900)--
(6.720751,4.367714)--
(6.808531,4.336306)--
(6.820688,4.331788)--
(6.901239,4.296383)--
(6.961842,4.268580)--
(7.058981,4.214224)--
(7.251736,4.082995)--
(7.761503,3.725993)--
(7.775895,3.715670)--
(7.810967,3.687060)--
(7.818614,3.680590)--
(7.862731,3.640187)--
(7.875293,3.626175)--
(7.924523,3.570102)--
(7.929217,3.563500)--
(7.951440,3.532224)--
(7.983109,3.484886)--
(8.011396,3.441241)--
(10.000000,0.000000)--
cycle;
%
%
\filldraw[fill=black!7]
(0.000000,0.000000)--
(0.013209,0.017785)--
(3.651627,3.366333)--
(3.929186,3.617911)--
(4.056146,3.705726)--
(4.205725,3.804987)--
(4.241652,3.826618)--
(4.288111,3.854181)--
(4.325734,3.876462)--
(4.341661,3.885236)--
(4.416739,3.924861)--
(4.537071,3.981879)--
(4.845931,4.122283)--
(5.009647,4.195642)--
(5.201251,4.274486)--
(5.303203,4.311319)--
(5.530592,4.392227)--
(5.705181,4.426388)--
(5.803504,4.445628)--
(5.849555,4.453645)--
(5.972342,4.465300)--
(6.002238,4.467864)--
(6.065618,4.471182)--
(6.176604,4.469357)--
(6.280715,4.463595)--
(6.451473,4.443636)--
(6.561077,4.417136)--
(6.595115,4.407274)--
(6.686541,4.379900)--
(6.720751,4.367714)--
(6.808531,4.336306)--
(6.820688,4.331788)--
(6.901239,4.296383)--
(6.961842,4.268580)--
(7.058981,4.214224)--
(7.251736,4.082995)--
(7.761503,3.725993)--
(7.775895,3.715670)--
(7.810967,3.687060)--
(7.818614,3.680590)--
(7.862731,3.640187)--
(7.875293,3.626175)--
(7.924523,3.570102)--
(7.929217,3.563500)--
(7.951440,3.532224)--
(7.983109,3.484886)--
(8.011396,3.441241)--
(7.142857,4.948717)--
(6.315789,5.469634)--
(5.869565,5.647992)--
(5.526316,5.697536)--
(5.384615,5.662474)--
(5.368421,5.651955)--
(3.333333,4.330127)--
(2.903226,3.911082)--
(2.500000,3.464102)--
(1.621622,2.340609)--
(1.000000,1.484615)--
(0.592593,0.898100)--
(0.340909,0.524864)--
(0.218819,0.341104)--
(0.049505,0.085745)--
cycle;

%% file: cdb.tex
\filldraw[fill=black!20]
(0.000000,0.000000)--
(0.516453,0.893712)--
(2.216103,3.833280)--
(2.483654,4.294813)--
(2.781357,4.806083)--
(3.137308,5.392704)--
(3.464911,5.928492)--
(3.724388,6.291442)--
(4.119200,6.793051)--
(4.241323,6.925489)--
(4.493695,7.126803)--
(4.670000,7.241399)--
(4.770526,7.302092)--
(4.919519,7.379744)--
(5.000000,7.407149)--
(5.107484,7.409396)--
(5.551129,7.374818)--
(5.803675,7.268249)--
(7.039569,5.127617)--
(10.000000,0.000000)--
cycle;
\filldraw[fill=black!7]
(5.000000,8.660254)--
(0.516453,0.893712)--
(2.216103,3.833280)--
(2.483654,4.294813)--
(2.781357,4.806083)--
(3.137308,5.392704)--
(3.464911,5.928492)--
(3.724388,6.291442)--
(4.119200,6.793051)--
(4.241323,6.925489)--
(4.493695,7.126803)--
(4.670000,7.241399)--
(4.770526,7.302092)--
(4.919519,7.379744)--
(5.000000,7.407149)--
(5.107484,7.409396)--
(5.551129,7.374818)--
(5.803675,7.268249)--
(7.039569,5.127617)--
cycle;

%% file: sidegd.tex
\filldraw[fill=black!20]
(0,0)--
(5.626,2.77)--
(5.774,2.82)--
(5.826,2.832)--
(5.918,2.846)--
(6.2,2.866)--
(7.12,2.784)--
(7.354,2.718)--
(7.922,2.388)--
(8.694,1.734)--
(9.474,0.912)--
(9.924,0.132)--
(10,0)--
cycle;
\filldraw[fill=black!7]
(0,0)--
(5.626,2.77)--
(5.774,2.82)--
(5.826,2.832)--
(5.918,2.846)--
(6.2,2.866)--
(7.12,2.784)--
(7.354,2.718)--
(7.922,2.388)--
(8.694,1.734)--
(9.474,0.912)--
(9.924,0.132)--
(10,0)--
(7.1429,4.9487)--
cycle;

%% file: sidedb.tex
\filldraw[fill=black!20]
(0,0)--
(0.076,0.132)--
(0.526,0.912)--
(1.102,1.55)--
(1.634,2.104)--
(1.82,2.246)--
(2.122,2.304)--
(3.058,2.474)--
(4.66,2.462)--
(5.158,2.422)--
(5.988,2.25)--
(6.168,2.186)--
(6.496,2.066)--
(7.654,1.534)--
(8.36,1.14)--
(8.526,1.04)--
(9.926,0.128)--
(10,0)--
cycle;
\filldraw[fill=black!7]
(0,0)--
(0.076,0.132)--
(0.526,0.912)--
(1.102,1.55)--
(1.634,2.104)--
(1.82,2.246)--
(2.122,2.304)--
(3.058,2.474)--
(4.66,2.462)--
(5.158,2.422)--
(5.988,2.25)--
(6.168,2.186)--
(6.496,2.066)--
(7.654,1.534)--
(8.36,1.14)--
(8.526,1.04)--
(9.926,0.128)--
(10,0)--
   (9.7812,0.3411)--
   (9.6591,0.5249)--
   (9.9505,0.0857)--
   (9.4074,0.8981)--
   (9.0000,1.4846)--
   (8.3784,2.3406)--
   (7.5000,3.4641)--
   (7.0968,3.9111)--
   (6.6667,4.3301)--
   (5.0000,5.4127)--
   (4.7368,5.4696)--
   (3.8462,5.3294)--
   (2.8571,4.9487)--
cycle;